\documentclass[10pt,conference,romanappendices]{IEEEtran}
\usepackage[cmex10]{amsmath}
\usepackage{stmaryrd}
\usepackage{amsfonts}
\usepackage{amssymb}
\usepackage{bbold}
\usepackage{wrapfig}
\usepackage{graphicx}
\usepackage{enumerate}
\usepackage{multirow}
\usepackage{xspace}
\usepackage{cancel}
\usepackage{enumitem}
\usepackage[hidelinks]{hyperref}
\usepackage{esvect}
\usepackage{caption}
\usepackage{cancel}
\newtheorem{example}{Example}
\newtheorem{definition}{Definition}
\newtheorem{proposition}{Proposition}
\newenvironment{proof}{\hspace{8pt}\ti{Proof:}}{$\blacksquare$}

\newtheorem{remark}{Remark}
\newcommand{\mprop}{Proposition~\ref{prop:sac_uns}\xspace}

\newcommand{\vic}[1]{$#1$-vicinity\xspace}
\newcommand{\lpr}[1]{$#1$-certificate\xspace}
\newcommand{\lprs}[1]{$#1$-certificates\xspace}

\newcommand{\imsat}{\mbox{$\mi{SAT}^{\mi{sa}}_{im}$}\xspace}
\newcommand{\sasat}{\mbox{$\mi{SAT}^{\mi{sa}}$}\xspace}

\newcommand{\sapqe}{\mbox{$\mi{PQE}^{\mi{sa}}$}\xspace}
\newcommand{\ppr}[1]{\mbox{$\mi{AppPr6}^*$}\xspace}

\newcommand{\slt}[2]{\mbox{$\mi{SatLits}(#1,\pnt{#2})$}\xspace}
\newcommand{\sslt}[2]{\mbox{$\mi{SatLits}(#1,\pent{#2}{^*})$}\xspace}

\newcommand{\dpqe}{\mbox{$\mi{DS}$-$\mi{PQE}$}\xspace}

\newcommand{\ods}[2]{\mbox{$#1 \rightarrow #2$}}

\newcommand{\oods}[2]{\mbox{\pnt{#1}~$\rightarrow #2$}}
\newcommand{\bm}[1]{{\mbox{\boldmath $#1$}}}
\newcommand{\Bm}[1]{{\boldmath $#1$}}

\newcommand{\di}[1]{\mbox{$\mi{Diam}(#1)$}\xspace}

\newcommand{\rl}[1]{\mbox{$\mi{Rel}_{#1}$}\xspace}

\newcommand{\imp}{\Rightarrow} 
\newcommand{\ie}{i.e.,\xspace}

\newcommand{\pnt}[1]{\mbox{$\vv{#1}$}\xspace}

\newcommand{\ppnt}[2]{\mbox{$\vv{#1}\!_{#2}$}}

\newcommand{\pent}[2]{\mbox{$\vv{#1}#2$}}

\newcommand{\cof}[2]{\mbox{$#1_{\vec{#2}}$}}

\newcommand{\V}[1]{\mbox{$\mathit{Vars}(#1)$}}
\newcommand{\Va}[1]{\mbox{$\mi{Vars}(\vec{#1})$}}

\newcommand{\s}[1]{\mbox{$\{#1\}$}}

\newcommand{\nGz}[2]{$G_{non-\{z\}}$}

\newcommand{\prr}[1]{\mi{Prev}(\boldsymbol{q})}

\newcommand{\mi}[1]{\mathit{#1}}
\newcommand{\ti}[1]{\textit{#1}}
\newcommand{\tb}[1]{\textbf{#1}}

\newcommand{\ttt}{\>\>\>}
\newcommand{\tttt}{\>\>\>\>}

\newcommand{\Tt}{\>\>}
\newcommand{\Sub}[2]{\mbox{$\mi{#1}\!_{\mi{#2}}$}}

\newcommand{\prob}[2]{\mbox{$\exists{#1} [#2]$}}
\newcommand{\Prob}[3]{\mbox{$\exists{#1}\exists{#2}[#3]$}}
\newcommand{\Comment}[1]{}

\newcommand{\cp}{\mbox{$\mi{CP}^{\mi{pqe}}$}\xspace}

\newcommand{\aabs}[2]{\mbox{$S_{#1,#2}$}}
\newcommand{\Abs}[1]{\mbox{$S_{1,#1}$}}

\begin{document}

%\title{Structure-Aware Computing By Partial Quantifier Elimination}
%\title{Structure-Aware Computing And Partial Quantifier Elimination}
\title{Structure-Aware Computing, Partial Quantifier Elimination And SAT}

\author{\IEEEauthorblockN{Eugene Goldberg} 
\IEEEauthorblockA{
email:
eu.goldberg@gmail.com}}

%% \author{Eugene Goldberg}
%% \institute{\email{eu.goldberg@gmail.com}}
%% \vspace{-30pt}
%% \author{}
%% \institute{}

\maketitle

\begin{abstract}
Typically, a practical algorithm of hardware verification obtains a
\ti{semantic} result by being applied to a \ti{particular} formula
$F$. That is, although this algorithm uses the specifics of $F$
(sometimes \ti{inadvertently}), its result holds for all formulas
logically equivalent to $F$. We refer to computations that get a
semantic result by \textit{intentionally} exploiting the specifics of
$F$ as \ti{structure-aware computing} (SAC). Arguably, using SAC
allows one to get a semantic result faster. We show that
\textit{partial quantifier elimination} (PQE), a generalization of
quantifier elimination, can be used for SAC. The objective of this
paper is twofold. First, we explain how SAC is performed by PQE by the
example of three verification problems (property generation,
equivalence checking and model checking). Second, we argue that PQE
solving \ti{itself} can benefit from SAC. We use SAT solving to
introduce a technique that can be also applied in structure-aware PQE
solving.
\end{abstract}

%\vspace{-2pt}
\section{Introduction}
\label{sec:intro}
A practical algorithm of hardware verification usually produces a
\ti{semantic} result by being applied to a \ti{particular} formula and
using its specifics (often inadvertently).  Suppose, for instance,
that a simple SAT solver performing the DPLL procedure~\cite{dpll}
proves a formula $F$ unsatisfiable. Although it produces a semantic
result (\ie this result holds for all formulas logically equivalent to
$F$), it also unwittingly uses the structure of $F$ via applying
BCP. We will say that an algorithm producing a semantic result
performs \tb{structure-aware computing} if, to achieve greater
efficiency, it \ti{intentionally} exploits the structure of the
formula at hand.  The objective of this paper is twofold. First, we
want to show that structure-aware computing can be done by partial
quantifier elimination (\tb{PQE}). Second, we want to substantiate the
point that PQE solving itself can benefit from greater
structure-awareness.

In this paper, we consider only \ti{propositional} formulas in
\ti{conjunctive normal form} (CNF) and only \ti{existential}
quantifiers. PQE is a generalization of regular quantifier elimination
(QE) defined as follows~\cite{hvc-14}.  Let $F(X,Y)$ be a
quantifier-free formula where $X,Y$ are sets of variables and $G$ be a
subset of clauses\footnote{Given a CNF formula $F$ represented as the conjunction of clauses
$C_1 \wedge \dots \wedge C_k$, we will also consider $F$ as
the \ti{set} of clauses \s{C_1,\dots,C_k}.
} of $F$.  Given a formula
\prob{X}{F}, the PQE problem is to find a quantifier-free formula
$H(Y)$ such that $\prob{X}{F}\equiv H\wedge\prob{X}{F \setminus G}$.
In contrast to \ti{full} QE, only the clauses of $G$ are taken out of
the scope of quantifiers hence the name \ti{partial} QE.  Note that QE
is just a special case of PQE where $G = F$ and the entire formula
gets unquantified.  A key role in PQE solving plays \ti{redundancy
  based reasoning}: to take a set of clauses $G$ out of
\prob{X}{F(X,Y)}, one essentially needs to find a formula $H(Y)$ that
makes $G$ \ti{redundant} in $H \wedge \prob{X}{F}$.  The appeal of PQE
is that it can be \ti{much more efficient} than QE if $G$ is a small
piece of $F$. To solve PQE, one needs to make redundant only $G$
whereas in QE the \ti{entire} formula $F$ is redundant in $H \wedge
\prob{X}{F}$.
%So, it is beneficial to design algorithms based on PQE.

The \tb{idea} of structure-aware computing by PQE is based on the
following observation. QE is a \ti{semantic} operation in the sense
that if $\prob{X}{F(X,Y)} \equiv H(Y)$ and $F' \equiv F$, then
$\prob{X}{F'} \equiv H$. That is, to verify the correctness of $H$ it
suffices to know the truth table of $F$. On the other hand, PQE is a
\ti{structural} (\ie formula-specific) operation. Assume, for
instance, that $F \equiv F'$ and both $F$ and $F'$ contain a subset of
clauses $G$.  The fact that $\prob{X}{F} \equiv H \wedge \prob{X}{F
  \setminus G}$ \ti{does not} imply that $\prob{X}{F'} \equiv H \wedge
\prob{X}{F' \setminus G}$.  For instance, $G$ may be redundant in
\prob{X}{F} (and so $H \equiv 1$ and $\prob{X}{F} \equiv \prob{X}{F
  \setminus G}$) but not in \prob{X}{F'}.  Thus, one cannot prove $H$
correct using the truth table of $F$ alone because the correctness of
$H$ depends on the particulars of $F$.  That is $H$ is
\ti{formula-specific}. In this sense, PQE is similar to
\ti{interpolation}~\cite{craig,ken03} that is a structural operation
too.

Besides introducing the idea of structure-aware computing by PQE our
contribution is as follows. \ti{First}, we show how structure-aware
computing by PQE works in two existing methods (property generation
and equivalence checking,
Sections~\ref{sec:prop_gen},\ref{sec:eq_check}) and a new method
performing an instance of model checking
(Section~\ref{sec:pqe_mc}). Although the previous experimental results
of~\cite{fmcad16,cav23} recalled in this paper show the viability of
existing PQE solvers, their performance can and should be drastically
improved. So, one of the goals of this paper is to emphasize the
importance of PQE to inspire building new efficient PQE
solvers. \ti{Second}, we show that PQE \ti{itself} can be made more
structure-aware (Section~\ref{sec:pqe_flaws}). \ti{Third}, we present
a technique for structure-aware SAT solving
(Sections~\ref{sec:sac_sat_gen}\,-\ref{sec:sac_sat_alg}). Due to the
tight relation between SAT and PQE, this technique can be used in
structure-aware PQE solving. One can also employ it for building
structure-aware SAT solvers tuned to particular applications.

The main body of this paper is structured as follows. (Some additional
information is given in the appendix.)  In Section~\ref{sec:basic}, we
provide basic definitions. A high-level view of PQE solving is
presented in Section~\ref{sec:pqe_alg}. Section~\ref{sec:interp} shows
that interpolation can be viewed as a special case of PQE. As
mentioned above, Sections~\ref{sec:prop_gen}\,-\ref{sec:sac_sat_alg}
discuss various aspects of structure-aware computing by PQE and
structure-aware SAT solving.  In Section~\ref{sec:concl}, we make
conclusions.
\clearpage

%\section{Redundant Clauses, Boundary Points, And Quantifier Elimination}
\section{Basic Definitions}
%\flushbottom
\label{sec:basic}

In this section, when we say ``formula'' without mentioning
quantifiers, we mean ``a quantifier-free formula''.

\begin{definition}
\label{def:cnf}
We assume that formulas have only Boolean variables.  A \tb{literal}
of a variable $v$ is either $v$ or its negation $\overline{v}$.  A
\tb{clause} is a disjunction of literals. A formula $F$ is in
conjunctive normal form (\tb{CNF}) if $F = C_1 \wedge \dots \wedge
C_k$ where $C_1,\dots,C_k$ are clauses. We will also view $F$ as the
\tb{set of clauses} \s{C_1,\dots,C_k}. We assume that \tb{every
  formula is in CNF} unless otherwise stated.
\end{definition}

%
%  Vars(F)
%
%\vspace{-10pt}
\begin{definition}
  \label{def:vars} Let $F$ be a formula. Then \bm{\V{F}} denotes the
set of variables of $F$ and \bm{\V{\prob{X}{F}}} denotes
$\V{F}\!\setminus\!X$.
\end{definition}

%
% Assignment
%
%\vspace{-10pt}
\begin{definition}
Let $V$ be a set of variables. An \tb{assignment} \pnt{q} to $V$ is a
mapping $V'~\rightarrow \s{0,1}$ where $V' \subseteq V$.  We will
denote the set of variables assigned in \pnt{q}~~as \bm{\Va{q}}. We will
refer to \pnt{q} as a \tb{full assignment} to $V$ if $\Va{q}=V$.
\end{definition}

%
% satisfied/falsified
%
\begin{definition}
A literal, a clause and a formula are \tb{satisfied} (respectively
\tb{falsified}) by an assignment \pnt{q} if they evaluate to 1
(respectively 0) under \pnt{q}.
\end{definition}

%
%  Cofactor
%
%\vspace{-10pt}
\begin{definition}
\label{def:cofactor}
Let $C$ be a clause. Let $H$ be a formula that may have quantifiers,
and \pnt{q} be an assignment to \V{H}.  Denote by \bm{\cof{C}{q}} the
clause $C$ under assignment \pnt{q}. If $C$ is satisfied by \pnt{q},
then $\cof{C}{q} \equiv 1$. Otherwise, \cof{C}{q} is the clause
obtained from $C$ by removing all literals falsified by
\pnt{q}. Denote by \bm{\cof{H}{q}} the formula obtained from $H$ by
removing the clauses satisfied by \pnt{q} and replacing every clause
$C$ unsatisfied by \pnt{q} with \cof{C}{q}.
\end{definition}

%
% X-clause, non-X-clause
%
%\vspace{-10pt}
\begin{definition}
  \label{def:Xcls}
Given a formula \prob{X}{F}, a clause $C$ of $F$ is called
\tb{quantified} if \V{C} $\cap\, X \neq \emptyset$. Otherwise, $C$ is
called \tb{unquantified}.
\end{definition}

%
%   Formula equivalence
%
%\vspace{-10pt}
\begin{definition}
\label{def:formula-equiv}
Let $F', F''$ be formulas that may have existential quantifiers. We say
that $F', F''$ are \tb{equivalent}, written \bm{F' \equiv F''}, if $\cof{F'}{q} =
\cof{F''}{q}$ for all full assignments \pnt{q} to $\V{F'} \cup \V{F''}$.
\end{definition}

%
% redundant clauses
%
%\vspace{-10pt}
\begin{definition}
\label{def:red_cls}
Let $F(X,Y)$ be a formula and $G \subseteq F$ and $G \neq
\emptyset$. The clauses of $G$ are said to be \textbf{redundant in}
\bm{\prob{X}{F}} if $\prob{X}{F} \equiv \prob{X}{F \setminus G}$. If
$F \setminus G$ implies $G$, the clauses of $G$ are redundant in
\prob{X}{F} but the reverse is not true.
\end{definition}

%
%   PQE problem
%
%\vspace{-10pt}
\begin{definition}
\label{def:pqe_prob} Given a formula \prob{X}{F(X,Y))} and $G$ where
 $G \subseteq F$, the \tb{Partial Quantifier Elimination} (\tb{PQE})
problem is to find $H(Y)$ such that \Bm{\prob{X}{F}\equiv
  H\wedge\prob{X}{F \setminus G}}.  (So, PQE takes $G$ out of the
scope of quantifiers.)  The formula $H$ is called a \tb{solution} to
PQE. The special case of PQE where $G = F$ is called \tb{Quantifier
  Elimination} (\tb{QE}).
\end{definition}

%
% example
%
\begin{example}
\label{exmp:pqe_exmp}
Consider formula $F = C_1 \wedge \dots \wedge C_5$ where
$C_1=\overline{x}_3 \vee x_4$, $C_2\!=\!y_1\!\vee\!x_3$, $C_3=y_1 \vee
\overline{x}_4$, $C_4\!=\!y_2\!\vee\!x_4$,
$C_5\!=\!y_2\!\vee\!\overline{x}_4$. Let $Y =\s{y_1,y_2}$ and $X =
\s{x_3,x_4}$. Consider the PQE problem of taking $C_1$ out of
\prob{X}{F} \ie finding $H(Y)$ such that $\prob{X}{F} \equiv H \wedge
\prob{X}{F \setminus \s{C_1}}$. In Subsection~\ref{ssec:pqe_exmp}, we
show that\linebreak$\prob{X}{F} \equiv y_1 \wedge \prob{X}{F \setminus
  \s{C_1}}$ \ie $H\!  =\!y_1$ is a solution to this PQE problem
$\blacksquare$
\end{example}
%
% Resolvable clauses
%
\begin{definition}
\label{def:res_of_cls}
Let clauses $C'$,$C''$ have opposite literals of exactly one variable
$w\!\in\!\V{C'}\!\cap\!\V{C''}$.  Then $C'$,$C''$ are called
\tb{resolvable} on~$w$.  Let $C$ be the clause consisting of the
literals of $C'$ and $C''$ minus those of $w$. Then $C$ is said to be
obtained by \tb{resolution} of $C'$ and $C''$ on $w$.
\end{definition}

%
% Definition: blocked clause
%
\begin{definition}
\label{def:blk_cls}
Let $C$ be a clause of a formula $F$ and $w \in \V{C}$. The clause $C$
is called \tb{blocked} in $F$ at $w$~\cite{blocked_clause} if no
clause of $F$ is resolvable with $C$ on $w$.
\end{definition}

%
% blocked clause is redundant
%
\begin{proposition}
\label{prop:blk_cls}
Let a clause $C$ be blocked in a formula $F(X,Y)$ at a variable $x \in
X$.  Then $C$ is redundant in \prob{X}{F}, \ie \prob{X}{F \setminus
  \s{C}} $\equiv$ \prob{X}{F}.
\end{proposition}

Proposition~\ref{prop:blk_cls} was proved in~\cite{cav23}. The proofs
of all \ti{new} propositions are given in Appendix~\ref{app:proofs}.

\section{PQE solving}
\label{sec:pqe_alg}
In this section, we briefly describe the PQE algorithm called
\dpqe~\cite{hvc-14} where '$\mi{DS}$' stands for 'D-sequent' (see below). Our
objective here is to provide an idea of how the PQE problem can be
solved. So, in Subsection~\ref{ssec:high_level}, we give a high-level
description of this algorithm and in Subsection~\ref{ssec:pqe_exmp} we
present an example of PQE solving.

\subsection{Some background}
Information on QE in propositional logic can be found
in~\cite{blocking_clause,cav09,cav11,cmu,nik2,cadet_qe}.  QE by
redundancy based reasoning is presented in~\cite{fmcad12,fmcad13}.
One of the merits of such reasoning is that it allows to introduce
\ti{partial} QE.  A description of PQE algorithms and their sources
can be found in~\cite{hvc-14,cav23,cert_tech_rep,ds_pqe,eg_pqe_plus}.
Although, as shown in~\cite{fmcad16,cav23}, the power of existing PQE
solvers is sufficient to demonstrate the viability of PQE, their
performance leaves much to be desired. One of the reasons for that is
discussed in Section~\ref{sec:pqe_flaws}.

%
% subsection
%
\subsection{High-level view}
\label{ssec:high_level}
Like all existing PQE algorithms, \dpqe uses \ti{redundancy based
  reasoning} justified by the proposition below.
%
% proposition
%
\begin{proposition}
\label{prop:sol_for_pqe}
  Formula $H(Y)$ is a solution to the PQE problem of taking 
  $G$ out of \prob{X}{F(X,Y)} (\ie $ \prob{X}{F} \equiv H \wedge
  \prob{X}{F \setminus G}$) iff
  \begin{enumerate}
  \item $F \imp H$ and
  \item $H \wedge
    \prob{X}{F} \equiv H \wedge \prob{X}{F \setminus G}$
  \end{enumerate}
\end{proposition}

So, to take $G$ out of \prob{X}{F(X,Y)}, it suffices to find $H(Y)$
implied by $F$ that makes $G$ \ti{redundant} in $H \wedge
\prob{X}{F}$. We refer to the clauses of $G$ as the \tb{target} ones.
Below, we provide some basic facts about \dpqe. Since taking out an
unquantified clause is trivial, we assume that $G$ contains only
\ti{quantified} clauses.  \dpqe finds a solution to the PQE problem
above by branching on variables of $F$.  The idea here is to reach a
subspace \pnt{q} where every clause of $G$ can be easily proved or
made redundant in \prob{X}{F}. \dpqe branches on unquantified
variables, \ie those of $Y$, \ti{before} quantified ones.  Like a
SAT-solver, \dpqe runs Boolean Constraint Propagation
(BCP)~\cite{dpll}. If a conflict occurs in subspace \pnt{q}, \dpqe
generates a conflict clause \Sub{C}{cnfl}~\cite{grasp} and adds it to
$F$ to \ti{make} clauses of $G$ redundant in subspace \pnt{q}. \dpqe
can also prove that $G$ is \ti{already} redundant in a subspace
\ti{without} reaching a conflict and adding a conflict clause.

To express the redundancy of a clause $C$ in \prob{X}{F} in a subspace
\pnt{q}, \dpqe uses a record \oods{q}{C} called a
\tb{D-sequent}~\cite{fmcad12,fmcad13}. This D-sequent also holds for
\prob{X}{F} after adding any clauses implied by $F$. A D-sequent
derived for a target clause $C$ is called \tb{atomic} if the
redundancy of $C$ can be trivially proved. D-sequents derived in
different branches can be resolved similarly to clauses. For
\ti{every} target clause $C$ of the original formula $G$, \dpqe uses
such resolution to eventually derive the D-sequent
\ods{\emptyset}{C}. The latter states the redundancy of $C$ in the
\ti{entire} space. Once the clauses of $G$ are proved redundant, \dpqe
terminates.  The solution $H(Y)$ to the PQE problem found by \dpqe
consists of the \ti{unquantified} clauses added to the initial formula
$F$ to make $G$ redundant.

%
% subsection
%
%\vspace{-6pt}
\subsection{An example of PQE solving}
\label{ssec:pqe_exmp}

Here we show how \dpqe solves
Example~\ref{exmp:pqe_exmp} introduced in Section~\ref{sec:basic}.
Recall that one takes $G = \s{C_1}$ out of \prob{X}{F(X,Y)} where $F =
C_1 \wedge \dots \wedge C_5$ and $C_1=\overline{x}_3 \vee x_4$,
$C_2\!=\!y_1\vee x_3$, $C_3=y_1 \vee \overline{x}_4$,
$C_4\!=\!y_2 \vee x_4$, $C_5\!=\!y_2 \vee\overline{x}_4$,
$Y=\s{y_1,y_2}$ and $X = \s{x_3,x_4}$.  One needs to find $H(Y)$ such
that $\prob{X}{F} \equiv H \wedge \prob{X}{F \setminus \s{C_1}}$.

Assume \dpqe picks the variable $y_1$ for branching and first explores
the branch $\pent{q}{'}=(y_1\!=\!0)$. In subspace \pent{q}{'}, clauses
$C_2,C_3$ become unit. (An unsatisfied clause is called \ti{unit} if
it has only one unassigned literal.)  After assigning $x_3\!=\!1$ to
satisfy $C_2$, the clause $C_1$ turns into unit too and a conflict
occurs: to satisfy $C_1$ and $C_3$, one has to assign the opposite
values to $x_4$. After a standard conflict analysis~\cite{grasp}, a
conflict clause $\Sub{C}{cnfl}=y_1$ is obtained by resolving $C_1$ and
$C_3$ on $x_4$ and resolving the resulting clause $y_1 \vee
\overline{x}_3$ with $C_2$ on $x_3$. To \ti{make} $C_1$ redundant in
subspace \pent{q}{'}, \dpqe adds \Sub{C}{cnfl} to $F$.  The redundancy
of $C_1$ is expressed by the D-sequent \ods{\pent{q}{'}}{C_1}.  This
D-sequent is an example of an \ti{atomic} one. It asserts that $C_1$
is redundant in subspace \pent{q}{'} because $C_1$ is implied by
another clause of $F$ (\ie \Sub{C}{cnfl}) in this subspace.

Having finished the first branch, \dpqe considers the second branch:
$\pent{q}{''}=(y_1=1)$. Since the clause $C_2$ is satisfied by
\pent{q}{''}, no clause of $F$ is resolvable with $C_1$ on variable
$x_3$ in subspace \pent{q}{''}. Hence, $C_1$ is blocked at variable
$x_3$ and, thus, redundant in \prob{X}{F} in subspace \pent{q}{''}.  So,
\dpqe generates the D-sequent \ods{\pent{q}{''}}{C_1}. This D-sequent
is another example of an \ti{atomic} D-sequent. It states that $C_1$
is \ti{already} redundant in \prob{X}{F} in subspace \pent{q}{''}
(without adding a new clause).
Then \dpqe resolves the D-sequents \ods{(y_1=0)}{C_1} and
\ods{(y_1=1)}{C_1} on $y_1$. This resolution produces the D-sequent
\ods{\emptyset}{C_1} stating the redundancy of $C_1$ in \prob{X}{F} in
the \ti{entire} space (\ie globally).  Recall that
$\Sub{F}{fin}=\Sub{C}{cnfl} \wedge \Sub{F}{init}$ where \Sub{F}{fin}
and \Sub{F}{init} denote the final and initial formula $F$
respectively. That is \Sub{C}{cnfl} is the only unquantified clause
added to \Sub{F}{init}. So, \dpqe returns \Sub{C}{cnfl} as a solution
$H(Y)$. The clause $\Sub{C}{cnfl} = y_1$ is indeed a solution since
$y_1$ is implied by \Sub{F}{init} and $C_1$ is redundant in $y_1
\wedge \prob{X}{\Sub{F}{init}}$. So both conditions of
Proposition~\ref{prop:sol_for_pqe} are met and, thus,
$\prob{X}{\Sub{F}{init}} \equiv y_1 \wedge \prob{X}{\Sub{F}{init}
  \setminus \s{C_1}}$.

%\vspace{-10pt}
\section{PQE And Interpolation}
\label{sec:interp}
In this section, we show that interpolation~\cite{craig,ken03} can be
viewed as a special case of PQE.

Let $A(X,Y) \wedge B(Y,Z)$ be an unsatisfiable formula where $X,Y,Z$
are disjoint sets of variables. Let $I(Y)$ be a formula such that $A
\wedge B \equiv I \wedge B$ and $A \imp I$. Replacing $A \wedge B$
with $I \wedge B$ is called \ti{interpolation} and $I$ is called an
\ti{interpolant}. PQE is similar to interpolation in the sense that
the latter is a \ti{structural} rather than semantic
operation. Suppose, for instance, that $A'(X,Y) \wedge B$ is a formula
such that $A' \wedge B \equiv A \wedge B$ but $A' \not\equiv A$. Then,
in general, the formula $I$ above \ti{is not} an interpolant for $A'
\wedge B$ \ie $A' \wedge B \not\equiv I \wedge B$ or $A' \not\imp I$.

Now, let us describe interpolation in terms of PQE. Consider the
formula \prob{W}{A \wedge B} where $W\!=\!X \cup Z$ and $A,B$ are the
formulas above. Let $A^*(Y)$ be obtained by taking $A$ out of
\prob{W}{A\!\wedge\!B} \ie $\prob{W}{A \wedge B} \equiv A^*\!\wedge
\prob{W}{B}$. Since $A \wedge B$ is unsatisfiable, $A^* \wedge B$ is
unsatisfiable too. So, \mbox{$A\wedge B \equiv A^* \wedge B$}. If $A
\imp A^*$, then $A^*$ is an interpolant $I$. If $A \not\imp A^*$, one
can extract an interpolant using the idea of~\cite{ken03} where $I(Y)$
is derived from a resolution proof that $A \wedge B$ implies an empty
clause. Similarly, in the case of PQE, one can extract an interpolant
from the resolution proof that $A \wedge B$ implies $A^*$. (It can be
provided by the PQE solver that produced $A^*$.)

The \ti{general case} of the PQE problem is different from the special
one above in three aspects. First, PQE is applied to the formula
\prob{W}{A(X,Y,Z) \wedge B(X,Y,Z)} where $A$ and $B$ depend on the
same set of variables. (In the special case above, $A$ \ti{does not}
depend on $Z$ and $B$ \ti{does not} depend on $X$.)  Second, $A \wedge
B$ can be \ti{satisfiable}.  Third, a solution $A^*$ is generally
implied by $A \wedge B$ rather than by $A$ \ti{alone}.  So,
interpolation can be viewed as a special case of PQE.

PQE allows one to build a compact ``interpolant'' $I(Y)$ even if the
formula $A(X,Y) \wedge B(Y,Z)$ above is \ti{satisfiable} (see
Appendix~\ref{app:interp}). Then the assignments satisfying $I \wedge
B$ are an \ti{abstraction} of those satisfying $A \wedge B$. Namely,
every (\pnt{y},\pnt{z}) satisfying $I \wedge B$ can be extended to
(\pnt{x},\pnt{y},\pnt{z}) satisfying $A \wedge B$ and vice versa. (A
trivial way to get $I(Y)$ is to perform full QE on \prob{X}{A}. But
then it is usually prohibitively large.)

%\vspace{-5pt}
\section{Structure-Aware Property Generation}
\label{sec:prop_gen}
In this section, we recall a method of property generation by
PQE~\cite{cav23}. We use this method as an example of structure-aware
computing. We explain property generation by the example of 
combinational circuits.
Let $M(X,V,W)$ be a combinational circuit where $X,V,W$ are sets of
internal, input, and output variables of $M$ respectively. Let
$F(X,V,W)$ denote a formula specifying $M$. As usual, this formula is
obtained by Tseitsin's transformations~\cite{tseitin}. Namely,
$F=F_{g_1} \wedge \dots \wedge F_{g_k}$ where $g_1,\dots,g_k$ are the
gates of $M$ and $F_{g_i}$ specifies the functionality of gate $g_i$.
%
% example
%
%\vspace{-3pt}
\begin{example}
\label{exmp:gate_cnf}
Let $g$ be a 2-input AND gate defined as $x_3 = x_1 \wedge x_2$ where
$x_3$ denotes the output variable and $x_1,x_2$ denote the input
variables. Then $g$ is specified by the formula
$F_g\!=\!(\overline{x}_1 \vee \overline{x}_2\vee x_3) \wedge (x_1 \vee
\overline{x}_3) \wedge (x_2 \vee \overline{x}_3)$. Every clause of
$F_g$ is falsified by an inconsistent assignment (where the output
value of $g$ is not implied by its input values). For instance,
$x_1\!\vee \overline{x}_3$ is falsified by the inconsistent assignment
$x_1\!=\!0, x_3\!=\!1$. So, every assignment \ti{satisfying} $F_g$
corresponds to a \ti{consistent} assignment to $g$ and vice
versa. Similarly, every assignment satisfying the formula $F$ above is
a consistent assignment to the gates of $M$ and vice versa
$\blacksquare$
\end{example}

For the sake of simplicity, consider generation of a property of the
circuit $M$ above by taking a single clause $C$ out of \prob{X}{F}.
Let $H(V,W) $ be a solution \ie $\prob{X}{F}\!\equiv\!H \wedge
\prob{X}{F \setminus \s{C}}$.  Since $F\!\imp\!H$, the solution $H$ is
a \ti{property}\footnote{In general, in addition to the internal
  variables $X$, input variables can be quantified too. In particular,
  if \ti{all} input variables are quantified, taking $C$ out of
  \Prob{X}{V}{F} produces a property $H(W)$ depending only on output
  variables. It states that if \pnt{w} falsifies $H$, then $M$ never
  produces the output \pnt{w}.} of $M$. If $H(V,W)$ is an
\ti{unwanted} property, $M$ has a bug. Every clause $Q(V,W)$ of $H$ is
implied by $F$ too and hence is a property of $M$ as well. (But $Q$ is
weaker than $H$.)

One can show that $Q$ specifies high-quality tests detecting the
``fault'' of $M$ simulated by removing $C$ from $F$.  Assume that a
full assignment $(\pnt{v},\pnt{w}^*)$ to $V \cup W$ falsifies
$Q$. Then \pnt{v} defines a test showing that the faulty version of
$M$ specified by $F \setminus \s{C}$ can produce a \ti{wrong} output
$\pnt{w}^*$ under the input \pnt{v} (see Appendix~\ref{app:stuck_at}
for details). This fault relates to a stuck-at fault~\cite{abram}. The
latter is an imaginary manufacturing bug where a signal line of $M$ is
stuck at 0 or 1.  Tests detecting stuck-at faults are known to be
structure-aware tests of very high quality. So, on the one hand,
$H(V,W)$ is a \tb{semantic} property holding for every circuit
logically equivalent to $M$. On the other hand, $H$ is produced by
\tb{structure-aware computing} (due to using PQE) and has a
\ti{special meaning} for $M$. Namely, $H$ specifies all tests
detecting the stuck-at fault of $M$ corresponding to the removal of a
single clause of $F$.

Below, to describe the status quo, we recall some experimental results
of~\cite{cav23}.  Those results were produced by an optimized version
of \dpqe. The latter was used to generate properties for the
combinational circuit $N_k$ obtained by unfolding a sequential circuit
$N$ for $k$ time frames. Each single property was generated by taking
out a clause from a formula specifying $N_k$ where all variables but
the state variables of the last time frame were quantified.  Those
properties of $N_k$ were employed to produce invariants of $N$. A
sample of HWMCC benchmarks containing from 100 to 8,000 latches was
used in those experiments. \dpqe managed to generate a lot of
properties of $N_k$ that turned out to be invariants of $N$. \dpqe
also successfully generated an \ti{unwanted} invariant of a FIFO
buffer and so identified a hard-to-find bug. However, the percentage
of PQE problems solved by \dpqe was relatively \ti{low}. (A single PQE
problem was to take a clause out of a quantified formula describing
$N_k$.)  So, the current methods of PQE solving need to be
\ti{significantly improved}.
%
%
%
%\newpage
%\vspace{-5pt}
%\input{old_files/e8q_check}

%% \subsection{Equivalence checking and structure-aware computing}
%% \label{ssec:ec_sac}
\section{Structure-Aware Equivalence Checking}
\label{sec:eq_check}
In this section, we discuss equivalence checking by PQE~\cite{fmcad16}
as one more example of structure-aware computing.  Let $M'(X',V',w')$
and $M''(X'',V'',w'')$ be the single-output combinational circuits to
check for equivalence.  Here $X^{\alpha},V^{\alpha}$ are the sets of
internal and input variables and $w^{\alpha}$ is the output variable
of $M^{\alpha}$ where $\alpha \in \s{'\,,\,''}$. Circuits $M',M''$ are
called \ti{equivalent} if they produce identical values of $w',w''$
for identical inputs $\pnt{v}',\pnt{v}''$ (\ie full assignments to
$V',V''$).

Let $F'(X'\!,V'\!,w')$ and $F''(X'',V'',w'')$ specify $M'$ and $M''$
respectively as described in Section~\ref{sec:prop_gen}. Let
$\bm{F^*}$ denote the formula $F'(X',V,w') \wedge F''(X'',V,w'')$
where $V'\!=\!V''\!=\!V$ since $M'$ and $M''$ must produce the same
outputs only for \ti{identical} inputs. If $M'$ and $M''$ are
\ti{structurally similar}, the most efficient equivalence checkers use
the method that we will call \tb{cut propagation
  (\tb{\ti{CP}})}~\cite{kuehlmann97,date01,berkeley}.  The idea of
\ti{CP} is to build a sequence of cuts $\mi{Cut}_1,\dots,\mi{Cut}_k$
of $M'$ and $M''$ and find cut points of $M',M''$ for which some
simple \ti{pre-defined} relations can be derived e.g functional
equivalence.  Let $\mi{Rel}_i$ denote the formula specifying the
relations found for $\mi{Cut}_i$. Computations move from inputs to
outputs where $\mi{Cut}_1\!=\!V$ and $\mi{Cut}_k\!=\!\s{w',w''}$. The
relation $\mi{Rel}_1$ is \ti{set} to the constant 1 whereas
$\mi{Rel}_2,\dots,\mi{Rel}_k$ are \ti{computed} in an inductive
manner.  That is $\mi{Rel}_i$ is obtained from $F^*$ and previously
derived $\mi{Rel}_1,\dots,\mi{Rel}_{i-1}$. The objective of \ti{CP} is
to prove $\mi{Rel}_k = (w' \equiv w'')$.  The main flaw of \ti{CP} is
that circuits $M'$ and $M''$ may not have cut points with pre-defined
relations even if $M'$ and $M''$ are very similar. In this case
\ti{CP} fails. So, it is \ti{incomplete} even for similar
circuits. Nevertheless, \ti{CP} is a successful practical
structure-aware method that exploits the similarity of $M'$ and $M''$.

\setlength{\intextsep}{4pt}
\begin{wrapfigure}{l}{1.9in}
%\begin{figure} 
 \begin{center}
    \includegraphics[width=1.6in]{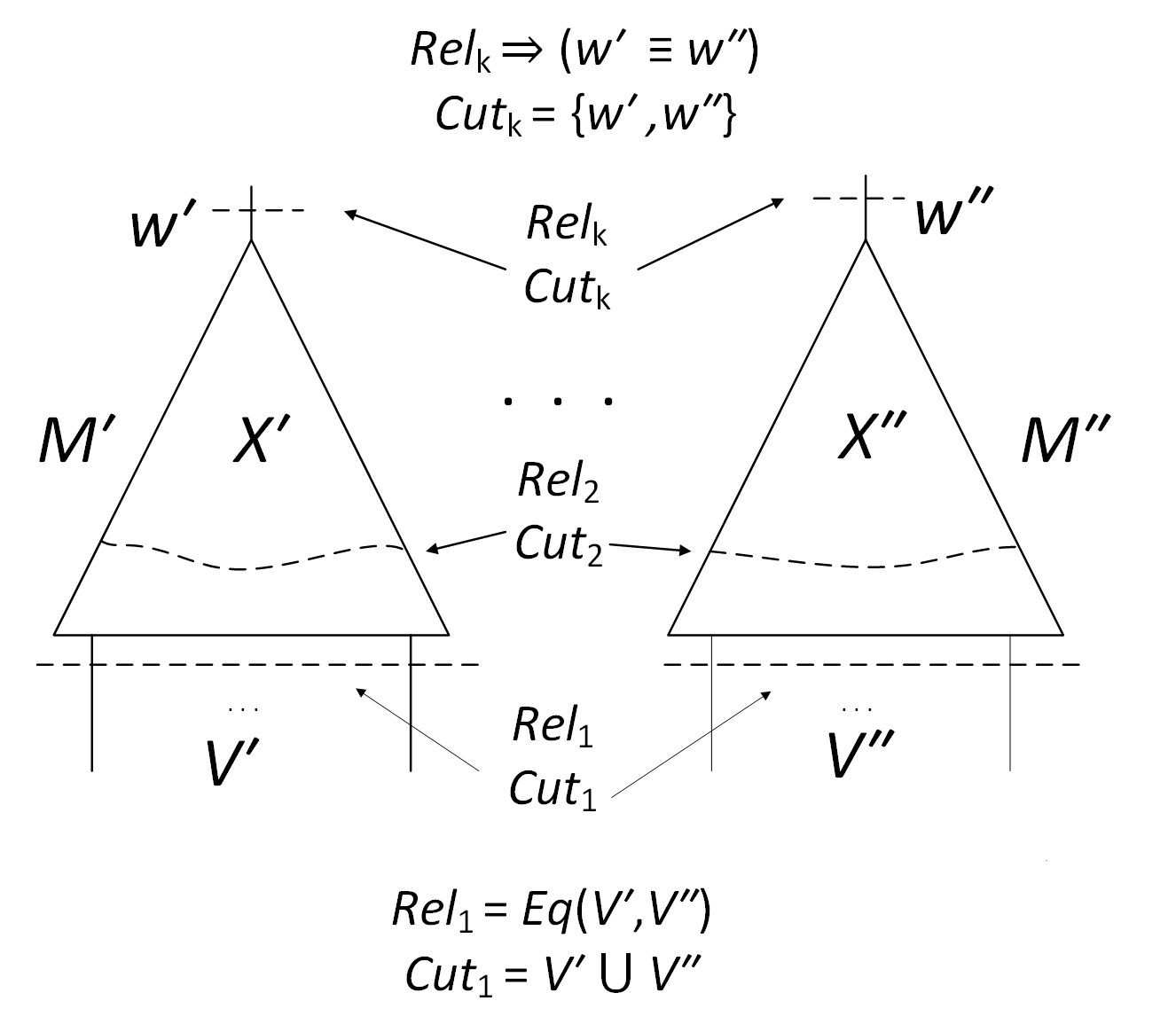}
  \end{center}
\vspace{-10pt}
\caption{Proving equivalence by\\ the \cp method}
\vspace{3pt}
\label{fig:cp}
%\end{figure}
\end{wrapfigure}

In~\cite{fmcad16}, a method of equivalence checking based on PQE was
introduced. This method also uses cut propagation. So, we will refer
to it as \cp. Let $\bm{F}$ denote the formula $F'(X',V',w') \wedge
F''(X'',V'',w'')$ where $V'$ and $V''$ are \ti{separate} sets. Let
$\bm{\mi{Eq}(V',V'')}$ denote a formula such that
\mbox{$\mi{Eq}(\pnt{v}\,',\pnt{v}\,'')$ = 1} iff $\pnt{v}\,' =
\pnt{v}\,''$.  Let $\bm{Z} = X' \cup X'' \cup V' \cup V''$.  \cp is based
on the proposition below proved in~\cite{fmcad16}.

\begin{proposition}
\label{prop:ec_by_pqe}
Assume $M',M''$ do not implement a constant (0 or 1).  Let
\prob{Z}{\mi{Eq} \wedge F} $\equiv H \wedge \prob{Z}{F}$.  Then $M'$
and $M''$ are equivalent iff $H \imp (w' \equiv w'')$.
\end{proposition}

Hence, to find if $M',M''$ are equivalent, it suffices to take \ti{Eq}
out of \prob{Z}{\mi{Eq} \wedge F}. Remarkably, PQE produces the same
result $H(w',w'')$ as \ti{full QE} on \prob{X}{\mi{Eq} \wedge F}. The
only exception, occurring when $M'$ and $M''$ are a constant, can be
ruled out by a few simple SAT checks. Note that $\mi{Eq} \wedge F$ is
semantically the same as $F^*$ above used in \ti{CP}. However, \cp
\ti{cannot} be formulated in terms of $F^*$ since it exploits the
\ti{structure} of $\mi{Eq} \wedge F$, namely, the presence of
$\mi{Eq}$.

\cp takes $\mi{Eq}$ out of \prob{Z}{\mi{Eq} \wedge F} incrementally,
cut by cut, by computing relations \rl{i} (see Fig.~\ref{fig:cp}). The
first cut $\mi{Cut}_1$ equals $V' \cup V''$ and \rl{1} is set to
$\mi{Eq}$. The remaining relations \rl{i} are computed by PQE (see
Appendix~\ref{app:ec_by_pqe}). Like \ti{CP}, the \cp method produces a
\tb{semantic} result (about equivalence of $M'$ and $M''$) by
\tb{structure-aware computing}. However, using PQE makes \cp much more
powerful. The difference between \cp and \ti{CP} is threefold. First,
\cp does not look for any \ti{no pre-defined relationships}.  (One can
show that \ti{CP} is just a special case of \cp where only particular
relationships between cut points are considered.)  Relations \rl{i}
computed by \cp just need to satisfy a very simple property:
$\prob{Z}{\rl{i-1} \wedge \rl{i} \wedge F} \equiv \prob{Z}{\rl{i}
  \wedge F}$. That is the next relation \rl{i} makes the previous
relation \rl{i-1} \ti{redundant}.  Second, \cp is
\ti{complete}. Third, as formally proved in~\cite{fmcad16}, relations
\rl{i} computed by \cp become quite \ti{simple} if $M'$ and $M''$ are
structurally similar.

Below, to describe the status quo, we recall some experiments with \cp
published in \cite{fmcad16}. In those experiments, circuits $M',M''$
containing a multiplier of various sizes were checked for
equivalence. (An optimized version of \dpqe was used as a PQE solver.)
$M',M''$ were intentionally designed so that they were structurally
similar but did not have any functionally equivalent cut points.  A
high-quality tool called ABC~\cite{abc} showed very poor performance,
whereas \cp solved all examples efficiently. At the same time, due to
insufficient power of \dpqe, the individual PQE problems of \cp had to
be solved \ti{approximately}\footnote{To boost PQE solving, a
  heuristic was used to drop some quantified clauses of the PQE
  problem at hand. That heuristic did not affect the correctness of
  the proof that $M'$ and $M''$ were equivalent but made the PQE
  solver return an under-approximation of a correct solution.}.  So,
building more powerful PQE solvers is of great importance.

%\input{s1at_by_pqe}
%\vspace{-5pt}
\section{Structure-Aware Model Checking}
%\vspace{-5pt}
\label{sec:pqe_mc}
In this section, we consider the problem of finding the reachability
diameter, \ie an instance of model checking. We describe a method of
solving this problem that serves as one more example of
structure-aware computing by PQE.
%
% subsection
%
%\vspace{-10pt}
\subsection{Motivation and some background}
An efficient algorithm for finding the reachability diameter can be
quite beneficial.  Suppose one knows that the reachability diameter of
a sequential circuit $N$ is less than $k$. Then, to verify \ti{any}
invariant of $N$, it suffices to check if it holds for the states of
$N$ reachable in at most $k\!-\!1$ transitions. This check can be done
by bounded model checking~\cite{bmc}. Finding the reachability
diameter of a sequential circuit by existing methods essentially
requires computing the set of all reachable
states~\cite{mc_book,mc_thesis}, which does not scale well.  An upper
bound on the reachability diameter called the recurrence diameter can
be found by a SAT-solver~\cite{rec_diam}. However, this upper bound is
very imprecise and its computing does not scale well either.
%
% Subsection
%
%\vspace{-10pt}
\subsection{Some definitions}
\label{ssec:defs}
Let $S$ specify the set of state variables of a sequential circuit
$N$.  Let $T(S',S'')$ denote the transition relation of $N$ where
$S',S''$ are the sets of present and next state variables.  Let
formula $I(S)$ specify the initial states of $N$. (A \tb{state} is a
full assignment to $S$.) A state \ppnt{s}{k+1} of $N$ with initial
states $I$ is called \tb{reachable} in $k$ transitions if there is a
sequence of states $\ppnt{s}{1},\dots,\ppnt{s}{k+1}$ such that
$I(\ppnt{s}{1})=1$ and $T(\ppnt{s}{i},\ppnt{s}{i+1})=1$,
$i=1,\dots,k$. For the reason explained in Remark~\ref{rem:stutter},
we assume that $N$ can \tb{stutter}.  That is, $T(\pnt{s},\pnt{s})=1$
for every state \pnt{s}.  (If $N$ lacks stuttering, it can be easily
introduced.)

%\vspace{-5pt}
\begin{remark}
  \label{rem:stutter}
If $N$ can stutter, the set of states of $N$ reachable in $k$
transitions is the same as the set of states reachable in \ti{at most}
$k$ transitions. Indeed, due to the ability of $N$ to stutter, each
state reachable in $p$ transitions is also reachable in $q$
transitions where $q > p$ $\blacksquare$
\end{remark}

%\vspace{-5pt}
Let $R_k(S)$ be a formula specifying the set of states of $N$
reachable in $k$ transitions. That is $R_k(\pnt{s})=1$ iff \pnt{s} is
reachable in $k$ transitions.  Let $S_i$ specify the state variables
of $i$-th time frame. Formula $R_k(S)$ where $S = S_{k+1}$ can be
computed by performing QE on \prob{\Abs{k}}{I_1 \wedge T_{1,k}}. Here
$\bm{\Abs{k}} = S_1 \cup \dots \cup S_{k}$ and $\bm{T_{1,k}} =
T(S_1,S_2) \wedge \dots \wedge T(S_k,S_{k+1})$. We will call \di{N,I}
the \tb{reachability diameter} of $N$ with initial states $I$ if any
reachable state of $N$ requires at most \di{N,I} transitions to reach
it.
%
%  subsection
%
%\vspace{-35pt}
%\subsection{Computing reachability diameter and  structure-aware computing}
\subsection{Computing the reachability diameter}
In this subsection, we consider the problem of deciding if \di{N,I} $<
k$.  A naive way to solve this problem is to compute $R_{k-1}$ and
$R_k$ by performing QE as described above.  $\di{N,I} < k$ iff
$R_{k-1} \equiv R_k$.  Unfortunately, computing $R_{k-1}$ and $R_k$
even for a relatively small value of $k$ can be very hard or simply
infeasible for large circuits.
Below, we show that, arguably, one can solve this problem more
efficiently by a \ti{structure-aware} procedure based on PQE.

\begin{proposition}
  \label{prop:rd_by_pqe}
Let $k \ge 1$. Let \prob{\Abs{k}}{I_1 \wedge I_2 \wedge T_{1,k}} be a
formula where $I_1$ and $I_2$ specify the initial states of $N$ in
terms of variables of $S_1$ and $S_2$ respectively.  Then $\di{N,I} <
k$ iff $I_2$ is redundant in \prob{\Abs{k}}{I_1 \wedge I_2 \wedge
  T_{1,k}}.
\end{proposition}

%\vspace{-4pt}
Proposition~\ref{prop:rd_by_pqe} reduces checking if $\di{N,I} < k$ to
the \ti{decision version} of PQE (\ie finding if $I_2$ is redundant in
\prob{\Abs{k}}{I_1 \wedge I_2 \wedge T_{1,k}}). Note that the presence
of $I_2$ simply ``cuts off'' the initial time frame (indexed by
1). So, \ti{semantically}, \prob{\Abs{k}}{I_1 \wedge I_2 \wedge
  T_{1,k}} is the same as \prob{\aabs{2}{k}}{I_2 \wedge T_{2,k}} \ie
specifies the states reachable in $k\!-\!1$ transitions. But the
former has a different \ti{structure} that can be exploited by
PQE. Namely, proving $I_2$ redundant in \prob{\Abs{k}}{I_1 \wedge I_2
  \wedge T_{1,k}} means that the sets of states reachable in $k\!-\!1$
and $k$ transitions (the latter specified by \prob{\Abs{k}}{I_1 \wedge
  T_{1,k}}) are identical.

Importantly, $I_2$ is a \ti{small} piece of the formula above. So,
proving it redundant can be much more efficient than computing
$R_{k-1}$ and $R_k$. For instance, computing $R_k$ by QE is equivalent
to proving the \ti{entire formula} $I_1 \wedge T_{1,k}$ redundant in
$R_k \wedge \prob{\Abs{k}}{I_1 \wedge T_{1,k}}$. However, making the
full use of such greater efficiency requires \ti{better PQE solvers}.

The reachability diameter is a \tb{semantic} notion in the sense that
it is identical for all logically equivalent circuits implementing the
same state transition graph. So, PQE allows one to obtain a semantic
result by \tb{structure-aware computing}.

%\section{Some Drawbacks Of Current PQE Algorithms}
\section{Making PQE Solving More Structure-Aware}
\label{sec:pqe_flaws}
%\subsection{Future exposition}
So far, we have considered using PQE for creating new methods of
structure-aware computing. However, PQE solving \ti{itself} can
benefit from being structure-aware. In this section, we describe a
problem with the current methods of PQE solving by the example of
\dpqe presented in Section~\ref{sec:pqe_alg}.  We argue that this
problem is caused by adding quantified conflict clauses (\ie the ones
depending on quantified variables) to the formula. Such clauses change
the structure of the formula by making its parts much more
``connected''. This problem can be addressed by using a technique
where one \ti{does not have to} add learned clauses to the formula. In
this paper, we introduce such a technique in the context of
SAT-solving
(Sections~\ref{sec:sac_sat_gen}-\ref{sec:sac_sat_alg}). Our future
goal is to incorporate this technique into PQE algorithms to make them
more structure-aware. (Appendix~\ref{app:sac_pqe} gives an idea of how
this can be done.)

%
%
%
%\subsection{Adding new quantified clauses}
%\subsection{Drawbacks caused by adding new quantified clauses}
%\vspace{3pt}
%\noindent
%\tb{Global and local targets}.
\subsection{Global and local targets}
\label{ssec:targets}
Before explaining the problem with adding conflict clauses, one needs
to describe the machinery of local targets used by \dpqe. Consider,
the PQE problem of taking a set of clauses $G$ out of
\prob{X}{F(X,Y)}. \dpqe does this by branching on variables of $F$
until the target clauses are proved/made redundant in the current
subspace \pnt{q}.  The clauses of $G$ are the \ti{global} targets
whose redundancy must be proved in the \ti{entire} space. But \dpqe
may also need to prove redundancy of some clauses of $F \setminus G$
\ti{locally} \ie in \ti{some} subspaces. For instance, when a global
target clause $C$ becomes unit in subspace \pnt{q}, \dpqe tries to
prove redundancy of the clauses resolvable with $C$ on the unassigned
variable. Denote the latter as $x$. (As mentioned earlier, \dpqe
assigns unquantified variables before quantified. So, $x$ is
quantified.)  If all these clauses are proved redundant in subspace
\pnt{q} \ie locally, $C$ is blocked in subspace \pnt{q} at $x$ and
hence redundant there. When a \ti{local} target clause $C'$ becomes
unit, its redundancy is proved in the same manner \ie via producing
\ti{new} local target clauses resolvable with $C'$. So, \ti{any}
quantified clause of $F \setminus G$ can become a local target.

%\noindent
%\tb{Strong and weak redundancy.}
\subsection{Strong and weak redundancy}
As we mentioned earlier, if a conflict occurs in subspace \pnt{q}, the
target clauses are \ti{made} redundant by adding a conflict clause
\Sub{C}{cnfl} that is falsified by \pnt{q}. (Otherwise, \dpqe simply
shows that the target clauses are already redundant in subspace
\pnt{q} without adding any clauses.) The addition of conflict clauses
is inherited by \dpqe from SAT solvers with clause
learning~\cite{grasp}. To explain the effect of adding conflict
clauses, let us consider the two types of clause redundancy below. A
clause $C$ is \ti{strongly} redundant in subspace \pnt{q} if
$\cof{F}{q} \setminus \s{\cof{C}{q}} \imp \cof{C}{q}$. Strong
redundancy means that $C$ is redundant in subspace \pnt{q} \ti{both}
in $F$ and \prob{X}{F}. A clause $C$ is \ti{weakly} redundant in
subspace \pnt{q} if it is redundant \ti{only} in \prob{X}{F} in this
subspace but not in $F$ (because $\cof{F}{q} \setminus \s{\cof{C}{q}}
\not\imp \cof{C}{q}$).

\dpqe identifies and exploits both types of clause redundancy above.
\tb{The problem} is that adding conflict clauses \ti{amplifies} the
ability to detect strong redundancies but \ti{diminishes} it for weak
ones. Indeed, on the one hand, adding conflict clauses increases the
number of implied assignments derived during BCP. If a target clause
$C$ becomes satisfied by the assignment derived from a unit clause
$C'$, it is strongly redundant in the current subspace (because $C$ is
implied by $C'$ in this subspace). So, the presence of conflict
clauses is beneficial. But this is \ti{not the case} for \ti{weak}
redundancies.

\dpqe checks the weak redundancy of a target $C$ when it becomes unit
in subspace \pnt{q}. Let $x$ be the unassigned variable of $C$.  As
mentioned above, to prove $C$ redundant in subspace \pnt{q}, \dpqe
treats the clauses resolvable with $C$ on $x$ as new local targets for
being proved redundant. If a quantified conflict clause is resolvable
with $C$ on $x$, it becomes a local target too. So, adding conflict
clauses can drastically increase the number of the local targets of
\dpqe and make proving weak redundancies \ti{much harder}. (This is
somewhat similar to the negative effect on so-called blocked sets of
clauses caused by adding new clauses obtained by
resolution~\cite{woody_allen}.)

%\newpage

%\section{SAT Solving By Structure-Aware Computing}
%\vspace{-5pt}
%\section{Structure-Aware SAT And PQE Solving}
\section{Structure-Aware SAT solving}
%\vspace{-5pt}
\label{sec:sac_sat_gen}
In this and the next two sections, we describe a technique called
\sasat facilitating structure-aware SAT-solving.  Our motivation is
twofold.  First, this technique can be used for fixing the problem
with PQE solving mentioned in the previous section
(Appendix~\ref{app:sac_pqe} has more details).  Second, \sasat is
important by itself since it, arguably, facilitates designing
structure-aware SAT solvers tuned to particular applications.

%\vspace{-5pt}
% Plan
% the reason for success of CDCL solvers
% a flaw of CDCL solvers by the example of equivalence checking
% a new direction to explore
%
% subsection
%
\subsection{Some background}
\label{ssec:backgr}
Modern CDCL SAT solvers are based on
resolution~\cite{grasp,resolution}. (\tb{\ti{CDCL}} stands for
\ti{conflict driven clause learning}.) A resolution proof that a
formula $F$ is unsatisfiable can be represented by a directed acyclic
graph $G$. Every node of $G$ corresponds to a clause.  The sources of
$G$ correspond to clauses of $F$ and its sink corresponds to an empty
clause. Every internal node of $G$ corresponds to a clause obtained by
resolving the clauses of the two incoming nodes.

Let a set of nodes of $G$ be a cut and \Sub{F}{cut} be the set of
clauses corresponding to the nodes of this cut. Note that the nodes
between this cut and the sink form a resolution proof that
\Sub{F}{cut} is unsatisfiable~\cite{novikov}. So one can build $G$ by
constructing a sequence of unsatisfiable formulas \Sub{F}{cut}
starting with the cut defined by the sources of $G$ and ending with
the cut specified by the sink. Importantly, producing \Sub{F}{cut}
corresponding to a cut makes every clause $C$ assigned to a node
\ti{below} this cut \ti{redundant}. (Because $\Sub{F}{cut} \imp
C$.) This nice property is implicitly used by the modern CDCL solvers
via decision-making introduced by Chaff~\cite{chaff} that is biased
towards recent conflict clauses. Such a bias is aimed at producing new
conflict clauses that make redundant (some) original clauses and conflict
clauses derived earlier.

So, the success of Chaff-like CDCL solvers
(e.g.~\cite{berkmin,minisat,picosat} and many others) can be
attributed to exploiting the basic structure of a \ti{resolution
  proof}. Unfortunately, this is not enough to exploit the fine
structure of the \ti{original formula}.  Assume, for instance, that
$F$ specifies equivalence checking of combinational circuits $M'$ and
$M''$.  As mentioned in Section~\ref{sec:eq_check}, if $M'$ and $M''$
are structurally similar, a short resolution proof of their
equivalence is obtained by deriving clauses relating their internal
points. However, CDCL solvers are not structure-aware enough to
produce it. For instance, if $M'$ and $M''$ are \ti{identical}
multipliers, proving $F$ unsatisfiable by a CDCL solver may take hours
despite the existence of a short resolution proof (that is linear in
$|F|$).

One can address the problem above by designing SAT solvers that can
exploit the structure of the original formula. The existing
structure-aware SAT-algorithms are typically CDCL-solvers with
heuristics tuned to a particular application (see
e.g.~\cite{circ_sat1,circ_sat2}).  SAT algorithms based on \sasat
should be quite different from CDCL solvers.  And as we argue below,
\sasat has powerful features that can make such algorithms much more
structure-aware.

%
% subsection
%
% Plan
% *) the general picture
% *) checking literal redundancy
% *) benefits of sasat
%    +) structure-awareness (decision making)
%    +) structure-awareness (clause separation and induction)
%    +) proof by induction and resolution
%
%\vspace{-5pt}
\subsection{A high-level view of \sasat}
\label{ssec:sasat_hilv}
Let $F(X)$ be the formula to check for satisfiability.  \sasat builds
a proof by performing local checks of literal redundancy.  Every such
a check finds out if a literal of a clause of $F$ is redundant in a
subspace \pnt{q}. Exploring literal redundancies has three
outcomes. First, a clause \Sub{B}{res} falsified by \pnt{q} is derived
by resolution. This case is somewhat similar to SAT solving by
CDCL. Second, after proving redundancy of particular literals for a
(small) subset of clauses, $F$ is proved unsatisfiable in subspace
\pnt{q} \ti{before} any clause \Sub{B}{res} can be derived.  This
proof is based on \ti{inductive reasoning} described in the next
section. In this case, \sasat builds an ``induction clause''
\Sub{B}{ind} falsified by \pnt{q} that confirms the unsatisfiability
of $F$ in subspace \pnt{q}. Third, if a literal of a clause of $F$ is
not redundant in subspace \pnt{q}, a satisfying assignment is
generated.

Here are some features of \sasat facilitating efficient
structure-aware computing.  \ti{First}, \sasat maintains \tb{two
  separate sets} of clauses: the set of clauses $F$ to check for
satisfiability and a set $P$ of learned clauses certifying literal
redundancies. When a new clause is learned, a SAT solver based on
\sasat is free to add it either to $F$ or to $P$. Only clauses of $F$
are checked for literal redundancy whereas both $F$ and $P$ are used
in BCP to learn new clauses. (In particular, one can add \ti{all}
learned clauses to $P$ keeping the original formula $F$ \ti{intact}.)
The \ti{intuition} here is that, in practice, short proofs usually
follow the structure of the original formula. By limiting the set of
clauses added to $F$, one can mostly preserve the original structure
of $F$, which facilitates finding a short proof.

\ti{Second}, the fact that decision making of \sasat is done on
\tb{clauses} (when checking literal redundancy) rather than variables,
makes it easier to exploit the structure of $F$. Suppose, for
instance, that $F$ specifies the equivalence checking of circuits
$M',M''$. A variable of $F$ is shared by \ti{many} gates but an
original clause of $F$ relates to a \ti{particular} gate of $M'$ or
$M''$ (Section~\ref{sec:prop_gen}). So, by picking clauses of $F$ in a
particular order, one can follow the fine structure of $M'$ and $M''$.

\ti{Third}, there is a reason to believe that resolution + induction
is \tb{more powerful} than pure resolution used by CDCL
solvers. Namely, the inductive reasoning of \sasat can be simulated by
adding blocked clauses (see Example~\ref{exmp:prop}) and resolution +
blocked clauses is exponentially more powerful than pure
resolution~\cite{gen_ext_resol}. Usually, converting a powerful
\ti{non-deterministic} proof system into an efficient
\ti{deterministic} algorithm is nearly impossible because the proof
space of such a system is huge. (For instance, a deterministic
algorithm has to look for ``good'' blocked clauses to add.)
Fortunately, for \sasat, tapping into the power of blocked clauses is
\ti{cheap} because inductive reasoning makes their \ti{explicit}
generation unnecessary.

%\vspace{-20pt}
%\section{Structure-Aware SAT Solving}
\section{Propositions Supporting \sasat}
\label{sec:sac_props}
In this section, we show how to prove a formula unsatisfiable by
combining literal redundancy checks with induction.

%\vspace{-5pt}
%
%   def: 
%
\begin{definition}
  \label{def:prim_sec}
Let $C$ be a clause of a formula $F$. We will denote by \bm{F_c} the
subset of clauses of $F$ containing $C$ and every clause sharing at
least one literal with $C$. We will refer to $F_c$ as the
\bm{C}-\tb{cluster} of clauses of $F$. We will call $C$ the
\tb{primary} clause and the clauses of $F_c \setminus \s{C}$ the
\tb{secondary} clauses of $F_c$.
\end{definition}
%
% example: C-cluster
%
\begin{example}
Let $F=C_1 \wedge C_2 \wedge C_3 \wedge C_4 \wedge \dots$ where $C_1 =
x_1 \vee x_2$, $C_2 = x_1 \vee \overline{x}_7 \vee x_9$, $C_3 = x_1
\vee \overline{x}_3$, $C_4 = x_2 \vee x_5 \vee x_6$. Assume
$C_1,\dots,C_4$ are the only clauses of $F$ with literals $x_1$ and/or
$x_2$. (But no restriction is imposed on clauses of $F$ with
$\overline{x}_1$ and/or $\overline{x}_2$.)  The clause $C_1$ specifies
the $C_1$-cluster $F_{c_1}=\s{C_1,C_2,C_3,C_4}$. Here $C_1$ is the
primary clause and $C_2,C_3,C_4$ are the secondary clauses of
$F_{c_1}$ $\blacksquare$
\end{example}

%
% def: vicinity
%
%\vspace{-7pt}
\begin{definition}
\label{def:vicin}
Let $C$ be a clause of $F(X)$ and $l$ be a literal of $C$.  The set of
full assignments to $X$ falsifying all literals of $C$ but $l$ is
called the \bm{l}\tb{-vicinity} of $C$. The shortest assignment
\pnt{q} satisfying $l$ and falsifying the other literals of $C$ is
said to \tb{specify} the \vic{l} of $C$. A clause $B$ is called an
\bm{l}\tb{-certificate} for $C$ if $F \imp B$ and \pnt{q} falsifies
$B$. 
\end{definition}

The existence of the certificate $B$ above means that the $l$-vicinity
of $C$ has no assignment satisfying $F$ and so, $l$ is \tb{redundant}
in $C$ (\ie can be dropped).
%
% example: literal vicinity
%
%\vspace{-3pt}
\begin{example}
  \label{exmp:vic}
Let $C\!=\!x_1 \vee x_2 \vee x_3$ be a clause of $F$. Then, say, the
$x_1$-vicinity of $C$ is specified by the assignment
\linebreak\pnt{q}\!=($\bm{x_1\!=\!1},x_2\!=\!0, x_3\!=\!0$). (The
assignment to $x_1$ satisfying $C$ is shown in bold.) The clause
$B\!=\!\overline{x}_1\!\vee x_2\!\vee\!x_3$ is an \lpr{x_1} for $C$ if
$F\!\imp\!B$. Then no assignment of subspace \pnt{q} satisfies $F$ and
the literal $x_1$ is redundant. Clauses $x_2\!\vee\!x_3$ and
$\overline{x}_1$ implied by $F$ are also \lprs{x_1} for $C$
$\blacksquare$
\end{example}

\vspace{-1pt} The two propositions below relate literal redundancy and
(un)satisfiability. They improve the result of~\cite{provinglocally}.
%
%  poposition about literal vicinity
%
%\vspace{-3pt}
\begin{proposition}
  \label{prop:lit_vic}
Let $C$ be a clause of a formula $F$. If $F$ is satisfiable, there
exists a clause $C'$ of the $C$-cluster and a literal $l$ of $C'$
shared with $C$ such that the $l$-vicinity of $C'$ has an assignment
satisfying $F$. (So, $l$ is not redundant in $C'$.)
\end{proposition}
%
% proposition about inductive reasoning
%
%\vspace{-3pt}
\begin{proposition}
  \label{prop:sac_uns}
Let $C$ be a clause of $F$.  Assume that for every clause $C'$ of the
$C$-cluster, every literal $l$ of $C'$ shared with $C$ is
redundant. Then $F$ is unsatisfiable.
\end{proposition}

From the viewpoint of \mprop, derivation of an empty clause by a
SAT-solver based on pure resolution \tb{is overkill}. Such a SAT
solver proves redundancy of \ti{any subset} of literals in \ti{every}
clause involved in this derivation.

If a literal $l$ is redundant in $C'$, there is an \lpr{l} for
$C'$. Let $P$ be a formula consisting of clauses certifying redundancy
of the literals mentioned in \mprop. Since, by definition, every
certificate is implied by $F$, then $F \imp P$. As we show in
Example~\ref{exmp:prop}, $P$ can be satisfiable. So, the proof of
unsatisfiability by \mprop is \tb{formula-specific}. Indeed, if some
formula $F$ implies a \ti{satisfiable} formula, this does not make $F$
unsatisfiable. (One can only claim that $F$ is falsified by every
assignment falsifying $P$.)  \mprop actually employs \tb{inductive
  reasoning}: the fact that $F$ is unsatisfiable in the part of the
search space falsifying $P$ implies the unsatisfiability
everywhere. If \mprop applies in a subspace \pnt{q}, a clause implied
by $F$ and falsified by \pnt{q} can be generated as a confirmation of
the unsatisfiability of $F$ in this subspace (see
Section~\ref{sec:sac_sat_alg}).
%\vspace{-2pt}

\begin{example}
\label{exmp:prop} This example demonstrates an application
of \mprop. Let $F$ be equal to $C_1 \wedge C_2 \wedge C_3 \wedge\dots$
where $C_1 = x_1 \vee \overline{x}_2$, $C_2 = x_1 \vee x_5$ and $C_3 =
\overline{x}_2 \vee \overline{x}_6 \vee x_8$.  Assume that
$C_1,C_2,C_3$ are the only clauses of $F$ with literals $x_1$ and/or
$\overline{x}_2$ whereas $F$ can contain any number of clauses with
literals $\overline{x}_1$ and/or $x_2$.  Consider the $C_1$-cluster
$F_{c_1}$ equal to \s{C_1,C_2,C_3} where $C_1$ is the primary clause
and $C_2,C_3$ are the secondary ones.
According to \mprop, to prove $F$ unsatisfiable, it suffices to show
that the single literals $x_1$ and $\overline{x}_2$ are redundant in
$C_1$ and the literals $x_1$ and $\overline{x}_2$ are redundant in
$C_2$ and $C_3$ respectively.  That is there is no assignment
satisfying $F$ in subspaces $(\bm{x_1\!=\!1},x_2\!=\!1)$,
$(x_1\!=\!0,\bm{x_2\!=\!0})$ and $(\bm{x_1\!=\!1},x_5\!=\!0)$ and
$(\bm{x_2\!=\!0},x_6\!=\!1,x_8\!=\!0)$.

Assume that $F \imp P$ where $P = B'_1 \wedge B''_1 \wedge B_2 \wedge
B_3$ and $B'_1 = \overline{x}_1 \vee \overline{x}_2$, $B''_1 = x_1
\vee x_2$, $B_2 = \overline{x}_1 \vee x_5$, $B_3 = x_2 \vee
\overline{x}_6 \vee x_8$. Here $B'_1,B''_1$ are $x_1$-certificate and
$\overline{x}_2$-certificate for $C_1$ and $B_2,B_3$ are
$x_1$-certificate and $\overline{x}_2$-certificate for $C_2,C_3$
respectively. Note that $P$ is satisfiable whereas $F$ is
unsatisfiable due to \mprop. $P$ remains satisfiable even when
conjoined with $F_{c_1}$. So, one cannot derive an empty clause from
$P \wedge F_{c_1}$ by resolution and needs to use \ti{the rest} of the
formula $F$. The complexity of this additional resolution derivation
can be quite high (see below).

The assignments satisfying $P\!\wedge\!F_{c_1}$ are
($x_1\!=\!1,x_2\!=\!0,$\linebreak$x_5\!=\!1,x_6\!=\!0$) and
($x_1\!=\!1,x_2\!=\!0,x_5\!=\!1,x_8\!=\!1$). They satisfy at least two
literals of $C_1$,$C_2$,$C_3$. These assignments are eliminated by
adding a single \ti{blocked} clause $K = \overline{x}_1 \vee x_2 \vee
\overline{x}_5$. (It is blocked in $F$ at $x_1$ since it is
unresolvable with $C_1$ and $C_2$ on $x_1$.) So, $P \wedge F_{c_1}
\wedge K$ is unsatisfiable and an empty clause can be derived from it
in a few resolutions.  However, the derivation of $K$ \ti{itself} by
pure resolution can be \ti{difficult}. (As mentioned before,
resolution + blocked clauses is exponentially more powerful than
resolution~\cite{gen_ext_resol}. So, derivation of even a single
blocked clause by resolution can be exponentially long.) Note that by
claiming that $F$ is unsatisfiable, \mprop, in a sense, uses the power
of the blocked clause $K$ \ti{implicitly} \ie without producing it and
adding to $F$ $\blacksquare$
\end{example}

The proposition below states that the set of literals one needs to
prove redundant to prove $F$ unsatisfiable can be smaller than that of
\mprop.

\begin{proposition}
  \label{prop:optim}
Let $C$ be a clause of a formula $F$ and $F_c$ denote the $C$-cluster.
Let $F'_c$ and $F''_c$ be two disjoint subsets of $F_c$ such that $C
\in F'_c$ and $F_c = F'_c \cup F''_c$. Assume that for every clause
$C' \in F'_c$ and for every literal $l$ of $C'$ shared with $C$ there
is an \lpr{l} implied by $F \setminus F''_c$. Then $F \setminus F''_c$
(and hence $F$) is unsatisfiable.
\end{proposition}

%\vspace{-5pt}
Informally, $F'_c$ is the subset of $F_c$ \ti{involved} in generation
of certificates. Reducing the checks of literal redundancy to the
clauses of $F'_c$ is especially important if an algorithm based on
\sasat adds learned clauses to $F$ (since the $C$-cluster can grow
large).  Proposition~\ref{prop:optim} shows that inductive reasoning
can be efficient \tb{even if} the $C$-cluster is \ti{very large}.

%\vspace{-15pt}
%\section{Description Of \sasat Procedure}
\section{An Illustrative Implementation Of \sasat}
\label{sec:sac_sat_alg}
Earlier, we introduced the technique called \sasat that can be used to
design structure-aware SAT solvers tuned to particular
applications. So far, we have been working on simple implementations
of \sasat in Python~\cite{python} trying to answer various emerging
questions (see Subsection~\ref{ssec:questions}). In this section, we
provide an illustrative implementation of \sasat that we will refer to
as \bm{\imsat}. Here '\ti{im}' stands for 'implementation'.  We give a
description of \imsat just to show how one can combine literal
redundancy checks and induction. (See Appendix~\ref{app:sac_sat_exmpl}
for an example of how \imsat works on a simple formula.)
%
%
%
%\vspace{-3pt}
\subsection{A description of \imsat}
Let $F$ be the formula to solve.  Like a CDCL solver, \imsat learns
new clauses. The difference is at least threefold here. First, \imsat
checks \ti{redundancy of literals} instead of just searching for a
satisfying assignment. Second, \imsat claims $F$ to be unsatisfiable
in a subspace \pnt{q} if \ti{induction} applies (\mprop). Third,
\imsat maintains two separate sets of clauses: the formula $F$ and a
set $P$ of certificate clauses.

%
% \dpqe procedure
%
\setlength{\intextsep}{4pt}
\setlength{\textfloatsep}{2pt}
\begin{wrapfigure}{L}{1.6in}
%\begin{figure}
%\begin{center}
%\small
\footnotesize
%\normalsize
\vspace{-5pt}
\begin{tabbing}
aa\=bb\= cc\= dd\= \kill
$\imsat(F,P,\pnt{q},C,l)$\{\\
\scriptsize{1}\> $(\mi{redund},\pnt{q})\!:=\!\mi{BCP}(F\!\wedge\!P,\pnt{q},C,l)$ \\
\scriptsize{2}\> if ($\mi{redund}$) \{\\
\scriptsize{3}\Tt  $\Sub{B}{res} := GenByRes(F \wedge P,\pnt{q},C,l)$ \\
\scriptsize{4}\Tt return($F,P,\Sub{B}{res},\mi{nil}$)\} \\
\scriptsize{5}\> if $(\mi{Satisf}(F,\pnt{q}))$ return($F,P,\mi{nil},\pnt{q}$) \\
\scriptsize{6}\> while ($\mi{True}$) \{ \\
\scriptsize{7}\Tt  $C,l := \mi{PickClsLit}(F,P,\pnt{q})$ \\
\scriptsize{8}\Tt  $(F,P,B,\pnt{s})\!:=\!\imsat(F,P,\pnt{q},C,l)$  \\[2pt]
\scriptsize{9}\Tt if $(\pnt{s} \neq \mi{nil})$ return($F,P,\mi{nil},\pnt{s}$)\\
\scriptsize{10}\Tt if $(\mi{Falsif}(B,\pnt{q}))$ return($F,P,B,\emptyset$) \\
\scriptsize{11}\Tt $F,P := \mi{AddCls}(F,P,B)$ \\
\scriptsize{12}\Tt $C' := \mi{Induct}(F,P,\pnt{q})$ \\
\scriptsize{13}\Tt if ($C' \neq \mi{nil}$) \{ \\
\scriptsize{14}\ttt  $\Sub{B}{ind}\!:=\!\mi{FormCls}(F,P,C',\pnt{q})$ \\
\scriptsize{15}\ttt    return($F,P,\Sub{B}{ind},\mi{nil}$)\}\}\} \\
%\tb{\scriptsize{1}}\>   \\
\end{tabbing} 
\vspace{-15pt}
%\caption{Top procedure}
\caption{\imsat}
\vspace{3pt}
\label{fig:top_proc_sat}
%\end{figure}
\end{wrapfigure}

%\subsection{Top procedure of \imsat}
\imsat calls itself recursively. A recursive call accepts the current
formula $F$, the current set $P$ of certificate clauses, an assignment
\pnt{q} specifying the current subspace, a clause $C \in F$ and a
literal $l$ of $C$ specifying the $l$-vicinity to explore (see
Fig.~\ref{fig:top_proc_sat}).  In the initial call,
$P\!=\!\emptyset,\pnt{q}\!=\!\emptyset$, $C=\mi{nil}$, $l = \mi{nil}$.
If $F$ is satisfiable in subspace \pnt{q}, then \imsat returns a
satisfying assignment. Otherwise, it returns an $l$-certificate for
the clause $C$ in subspace \pnt{q}. Besides, \imsat returns the
(modified) formulas $F$ and $P$.

\imsat first runs BCP on $F\!\wedge\!P$ to make assignments specifying
the $l$-vicinity of $C$ (line 1). If BCP shows that $l$ is redundant,
a $l$-certificate \Sub{B}{res} is built by resolution and then
returned by \imsat (lines 2-4). For instance, $l$ is redundant if BCP
runs into a \ti{conflict}, see
Example~\ref{exmp:cnfl_cls}. (Appendix~\ref{app:red_by_bcp} lists two
cases of redundancy of $l$ that are identified by BCP \ti{before} any
conflict occurs.) If \pnt{q} extended by BCP satisfies $F$, then $l$
is not redundant and \imsat terminates returning \pnt{q} (line 5).

If \imsat does not terminate after BCP, it runs a \ti{while} loop
where literal redundancies of clauses of $F$ are checked (lines
6-15). The loop starts with picking a clause $C$ and an unassigned
literal $l$ of $C$ not proved redundant in subspace \pnt{q} yet (line
7). Then \imsat is recursively called to check the satisfiability of
$F$ in subspace \pnt{q} in the $l$-vicinity of $C$ (line 8). This call
returns either a clause $B$ that is an $l$-certificate for $C$ in
subspace \pnt{q} (so $l$ is redundant) or an assignment \pnt{s}
satisfying $F$ ($l$ is not redundant).

If a satisfying assignment is found, \imsat terminates\linebreak (line
9). \imsat also terminates if $B$ \ti{itself} is falsified by \pnt{q}
(line 10). (This case simulates regular SAT solving by CDCL.)
Otherwise, $B$ is added to $F$ \ti{or} $P$ (line 11).  In either case,
$B$ is used in BCP. But if $B$ is added to $P$ instead of $F$, it
\ti{is not present} in any $C'$-cluster where $C' \in F$ even if $B$
shares literals with $C'$. So, literals of $B$ are \tb{never checked}
for redundancy.  Finally, \imsat checks if proving $l$ redundant in
$C$ and some previous proofs of literal redundancy have enabled the
induction specified by \mprop for a $C'$-cluster (line 12). If such a
$C'$-cluster is found, \imsat generates an induction clause
\Sub{B}{ind} falsified by \pnt{q} and terminates (lines
13-15). Generation of \Sub{B}{ind} is explained in
Example~\ref{exmp:ind_cls} and formally described in
Appendix~\ref{ssec:cnfl_cls_props}.

\begin{example}
 \label{exmp:cnfl_cls}
Here is an example of literal redundancy identified by BCP when a
conflict occurs. Let $F = C_1 \wedge C_2 \wedge C_3 \wedge \dots$
where $C_1 = x_1 \vee x_2 \vee x_3$, $C_2 = x_1 \vee x_2 \vee x_4$,
$C_3 = x_2 \vee \overline{x}_4$.  Assume that \imsat needs to check
the redundancy of $x_3$ in $C_1$ and the current assignment \pnt{q} is
empty. Then \imsat makes a recursive call with $C=C_1$ and $l = x_3$.
It explores the subspace $(x_1=0,x_2=0,\bm{x_3=1})$ specifying the
$x_3$-vicinity of $C_1$.  Making the decision assignment $x_1=0$ does
not produce any unit clauses (assuming that only clauses $C_1$ and
$C_2$ of $F$ have literal $x_1$) but adding the decision assignment
$x_2\!=\!0$ does. Namely, BCP derives $x_3\!=\!1$ from $C_1$ and the
conflicting assignments $x_4\!=\!1$ and $x_4\!=\!0$ from $C_2$ and
$C_3$ \ie runs into a \ti{conflict}. By resolving $C_2$ and $C_3$, one
produces the conflict clause $x_1 \vee x_2$ that serves as an
\lpr{x_3} for $C_1$.
\end{example}

\begin{example}
  \label{exmp:ind_cls} Here we illustrate the generation of a
clause \Sub{B}{ind} by \imsat when \ti{induction} applies in a
\tb{subspace}. Let $F$ be equal to $C_1 \wedge \dots \wedge C_4 \wedge
\dots$ where $C_1\!=\!\overline{x}_1\vee x_2\vee x_3$,
$C_2\!=\!\overline{x}_1 \vee x_5 \vee x_7$, $C_3\!=\!x_2 \vee
\overline{x}_6 \vee x_8$ and $C_4 = x_3 \vee x_9$. Assume that only
clauses $C_1,\dots,C_4$ of $F$ have literals $\overline{x}_1$, $x_2$
and $x_3$. Consider the $C_1$-cluster in subspace
$\pnt{q}\!=\!(x_1\!=\!1,x_4\!=\!0$,
\linebreak$x_9\!=\!1,x_{10}=0,x_{11}=0)$ \ie $C_1$ is the primary
clause. To apply \mprop in subspace \pnt{q}, one can ignore the
literals of $C_1$ falsified by \pnt{q} and discard the clauses
satisfied by \pnt{q} (see Proposition~\ref{prop:cnfl_subsp1} of
Appendix~\ref{app:proofs}). The literal $\overline{x}_1$ of $C_1$ is
falsified by \pnt{q} and $C_4$ is satisfied by \pnt{q}. So, it
suffices to check the redundancy of $x_2$ and $x_3$ in $C_1$ and $x_2$
in $C_3$.

Assume that \imsat derived certificates $B'_1 = x_3 \vee x_{10}$,
$B''_1 = x_2 \vee x_4$, $B_3= x_4 \vee \overline{x}_6 \vee x_8$.  Here
$B'_1$ and $B''_1$ are \lpr{x_2} and $x_3$-certificate for $C_1$ in
subspace \pnt{q} and $B_3$ is an $x_2$-certificate for $C_3$ in
subspace \pnt{q}. Then \mprop applies in subspace \pnt{q} and \imsat
generates the clause\linebreak $\Sub{B}{ind}=\overline{x}_1 \vee x_4
\vee \overline{x}_9 \vee x_{10}$ that is implied by $F$ and falsified
by \pnt{q}. The clause \Sub{B}{ind} consists of three parts (see
Proposition~\ref{prop:cnfl_cls1} of Appendix~\ref{app:proofs}). First,
it includes the literals of the primary clause $C_1$ falsified by
\pnt{q}. In our case, it is $\overline{x}_1$. Second, for every clause
of the $C_1$-cluster satisfied by \pnt{q}, \Sub{B}{ind} contains the
negation of a satisfied literal. In our case, it is the negation of
$x_9$ from $C_4$.  Finally, for each \lpr{l} above showing that the
literal $l$ is redundant in a cluster clause $C_i$, \Sub{B}{ind}
includes all the literals of this certificate but those of $C_i$.  For
instance, $B_3= x_4 \vee \overline{x}_6 \vee x_8$ is the
$x_2$-certificate produced for $C_3 = x_2 \vee \overline{x}_6 \vee
x_8$ in the subspace \pnt{q}.  So, \Sub{B}{ind} includes the literal
$x_4$ of $B_3$ since it is not in $C_3$ $\blacksquare$
\end{example}

\subsection{A few open questions related to \sasat (out of many)} 
\label{ssec:questions}
In this subsection, we list some questions that have not been
addressed in the description of \imsat.
One obvious question here is: is there a good \ti{general} heuristic
for picking the next clause $C \in F$ and a literal $l$ of $C$ to
check for redundancy? And how does one form a \ti{particular}
heuristic tuned to a class of formulas? Now, suppose $C$ and $l$ are
chosen in the current subspace \pnt{q}. What does one do with the
clauses previously derived to certify the redundancy of $l$ in $C$ in
subspaces visited earlier?  Should one extend \pnt{q} to satisfy all
those certificate clauses before examining the $l$-vicinity of $C$ in
the current subspace?  This would be similar to the decision-making of
Chaff~\cite{chaff}.  Another important question to answer is how one
decides whether a newly learned clause is added to the formula $F$ or
to the set of certificates $P$. How does this decision vary depending
on the structure of the formula?

Note that the questions above are not just technicalities. Answering
them requires a better understanding of the nature of structure-aware
SAT solving offered by \sasat.

%\input{r1eusing_pqe}
%\input{b1ackground}
%\vspace{-10pt}
\section{Conclusions}
\label{sec:concl}
%\vspace{-10pt}
%\input{e3g_pqe.fig}
All practical algorithms of hardware verification use the structure of
the formula at hand (often \ti{inadvertently}).  We show that
\ti{intentional} structure-aware computing (SAC) can be done by
partial quantifier elimination (PQE).  Interpolation, an early example
of SAC, can be viewed as a special case of PQE. We demonstrate that
SAC by PQE enables powerful verification algorithms. We also show that
PQE \ti{itself} can be made more structure-aware. As a first step to
achieving this goal, we introduce a technique that facilitates the
design of structure-aware SAT algorithms. Arguably, this technique can
be used for structure-aware PQE solving. Our discussion and results
suggest that PQE can be successfully used in the design of powerful
structure-aware algorithms.

%\clearpage
%\bibliographystyle{plain}
\bibliographystyle{IEEEtran}
\bibliography{short_sat,local,l1ocal_hvc}
%\newpage
%\clearpage
\appendices
%\vspace{4pt}
%% \appendix
%% \noindent{\large \tb{Appendix}}
%\input{e2xmp_of_sac}
%\vspace{-10pt}
\section{Interpolation For Satisfiable Formulas}
%\vspace{-5pt}
\label{app:interp}
Traditionally, interpolation is introduced via the notion of
implication. Suppose that a formula $A(X,Y)$ implies a formula
$\overline{B(Y,Z)}$ where $X,Y,Z$ are disjoint sets of variables. Then
there exists a formula $I(Y)$ called an interpolant such that $A \imp
I$ and $I \imp \overline{B}$. In terms of satisfiability this means
that $A \wedge B \equiv I \wedge B \equiv 0$ \ie $A \wedge B$ and $I
\wedge B$ are unsatisfiable. Interpolation can be used to simplify the
explanation why the implication $A \imp \overline{B}$ holds (by
replacing it with $I \imp \overline{B}$) or why $A \wedge B$ is
unsatisfiable (by replacing it with $I \wedge B$).

Let $A \not\imp \overline{B}$ \ie $A$ \tb{does not} imply
$\overline{B}$ and hence the formula $A \wedge B$ is \tb{satisfiable}.
Suppose we want to build an interpolant $I(Y)$ \ie a formula that
simplifies explaining why $A \not\imp \overline{B}$ and $A \wedge B
\not\equiv 0$ (by replacing $A$ with $I$). One cannot just mimic the
approach above by looking for a formula $I$ satisfying\linebreak $A
\imp I$ and $I \not\imp \overline{B}$. The problem here is that even
$I \equiv 1$ satisfies these implications for arbitrary $A$ and $B$
for which $A \not\imp \overline{B}$ holds.

Fortunately, one can build $I$ using PQE similarly to the case where
$A \wedge B$ is unsatisfiable (see Section~\ref{sec:interp}). Namely,
by taking $A$ out of the scope of quantifiers in \prob{W}{A \wedge B}
where $W = X \cup Z$. Let $A^*(Y)$ be a solution to this PQE problem
\ie $\prob{W}{A \wedge B} \equiv A^* \wedge \prob{W}{B}$. If $A \imp
A^*$, then $A^*$ is an interpolant in the following sense: every
assignment (\pnt{y},\pnt{z}) satisfying $A^* \wedge B$ can be extended
to (\pnt{x},\pnt{y},\pnt{z}) satisfying $A \wedge B$ and vice versa.

Indeed, if (\pnt{x},\pnt{y},\pnt{z}) satisfies $A(X,Y) \wedge B(Y,Z)$,
then it satisfies $B$ and $A^*$. (The latter holds, because $A \imp
A^*$.)  So, (\pnt{x},\pnt{y},\pnt{z}) satisfies $A^* \wedge B$ and
hence (\pnt{y},\pnt{z}) satisfies it too. On the other hand, if
(\pnt{y},\pnt{z}) satisfies $A^* \wedge B$, then it satisfies
$B$. Besides, there must be an \pnt{x} such that (\pnt{x},\pnt{y})
satisfies $A$. (Otherwise, $\prob{W}{A \wedge B} \neq A^* \wedge
\prob{W}{B}$ under assignment \pnt{y}.) Then (\pnt{x},\pnt{y},\pnt{z})
satisfies $A \wedge B$.

The relation between satisfying assignments of $A \wedge B$ and $A^*
\wedge B$ justifies using $A^*$ to explain why $A \wedge B \not\equiv
0$ and\linebreak $A \not\imp \overline{B}$ (by replacing $A$ with
$A^*$).  As mentioned earlier, if $A \not\imp A^*$, one can use the
technique of~\cite{ken03} to extract an interpolant from a resolution
derivation of $A^*$. The latter can be generated by the PQE solver
that produced $A^*$.
%\newpage

%\vspace{-20pt}
\section{Proofs Of New Propositions}
\setcounter{proposition}{0}
 \label{app:proofs}

%\subsection{Proofs of Section~\ref{sec:basic}}
%
%  proof of redundancy of a blocked clause
%
\stepcounter{proposition}
%\clearpage
%\vspace{-7pt}
\subsection{Propositions from  Section~\ref{sec:pqe_alg}}
%
%  properties of a solution for PQE
%
\begin{proposition}
%\label{prop:sol_for_pqe}
  Formula $H(Y)$ is a solution to the PQE problem of taking 
  $G$ out of \prob{X}{F(X,Y)} (\ie $ \prob{X}{F} \equiv H \wedge
  \prob{X}{F \setminus G}$) iff
  \begin{enumerate}
  \item $F \imp H$ and
  \item $H \wedge
    \prob{X}{F} \equiv H \wedge \prob{X}{F \setminus G}$
  \end{enumerate}
\end{proposition}
\begin{proof}
 \noindent\tb{The ``if'' part.} Assume conditions 1 and 2 hold. Since
 $H$ depends only on $Y$, then $H \wedge \prob{X}{F} \equiv \prob{X}{H
   \wedge F}$. Since $F \imp H$, then $\prob{X}{H \wedge F} \equiv
 \prob{X}{F}$. Then the second condition entails that $\prob{X}{F}
 \equiv H \wedge \prob{X}{F \setminus G}$.

\noindent\tb{The ``only if'' part.} Assume $\prob{X}{F} \equiv H
\wedge \prob{X}{F \setminus G}$. Let us show that conditions 1 and 2
hold.  Assume condition 1 fails \ie $F \not\imp H$. Then there is an
assignment (\pnt{x},\pnt{y}) satisfying $F$ and falsifying $H$. This
means that $\prob{X}{F} \neq$ \linebreak $H \wedge \prob{X}{F
  \setminus G}$ under assignment \pnt{y} and we have a
contradiction. To prove condition 2, one can simply conjoin both sides
of the equality $ \prob{X}{F} \equiv H \wedge \prob{X}{F \setminus G}$
with $H$
\end{proof}
%
%
%\subsection{Proofs of Section~\ref{sec:prev_res}}
%
%  proof of a propistion about equivalence checking
%

\stepcounter{proposition}

\subsection{Propositions from Section~\ref{sec:pqe_mc}}
%
%  reachability diameter
%
\begin{proposition}
%  \label{prop:rd_by_pqe}
Let $k \ge 1$. Let \prob{\Abs{k}}{I_1 \wedge I_2 \wedge T_{1,k}} be a
formula where $I_1$ and $I_2$ specify the initial states of $N$ in
terms of variables of $S_1$ and $S_2$ respectively.  Then $\di{N,I} <
k$ iff $I_2$ is redundant in \prob{\Abs{k}}{I_1 \wedge I_2 \wedge
  T_{1,k}}.
\end{proposition}
\begin{proof}
\noindent\tb{The ``if'' part.} Recall that $\Abs{k} = S_1 \cup \dots
\cup S_k$ and $T_{1,k} = T(S_1,S_2) \wedge \dots \wedge
T(S_k,S_{k+1})$. Assume that $I_2$ is redundant in \prob{\Abs{k}}{I_1
  \wedge I_2 \wedge T_{1,k}} \ie \prob{\Abs{k}}{I_1 \wedge T_{1,k}}
$\equiv$ \prob{\Abs{k}}{I_1 \wedge I_2 \wedge T_{1,k}}.  The formula
\prob{\Abs{k}}{I_1 \wedge T_{1,k}} is logically equivalent to $R_k$
specifying the set of states of $N$ reachable in $k$ transitions. On
the other hand, \prob{\Abs{k}}{I_1 \wedge I_2 \wedge T_{1,k}} is
logically equivalent to \prob{\aabs{2}{k}}{I_2 \wedge T_{2,k}}
specifying the states reachable in $k\!-\!1$ transitions (because
$I_1$ and $T(S_1,S_2)$ are redundant in \prob{\Abs{k}}{I_1 \wedge I_2
  \wedge T_{1,k}}). So, redundancy of $I_2$ means that $R_{k-1} \equiv
R_k$ and hence, $\di{N,I} < k$.

\vspace{3pt}
\noindent\tb{The ``only if'' part.} Assume the contrary \ie the
condition $\di{N,I} < k$ holds but $I_2$ is \ti{not} redundant in
\prob{\Abs{k}}{I_1 \wedge I_2 \wedge T_{1,k}}. Then there is an
assignment \linebreak $\pnt{p}=(\ppnt{s}{1},\dots,\ppnt{s}{k+1})$ such
that a) \pnt{p} satisfies $I_1\, \wedge\, T_{1,k}$;\linebreak b)
formula $I_1 \wedge I_2 \wedge T_{1,k}$ is unsatisfiable in subspace
\ppnt{s}{k+1}. (So, $I_2$ is not redundant because removing it from
$I_1 \wedge I_2 \wedge T_{1,k}$ makes the latter satisfiable in
subspace \ppnt{s}{k+1}.)  Condition 'b)' means that \ppnt{s}{k+1} is
unreachable in $k\!-\!1$ transitions whereas condition 'a)' implies
that \ppnt{s}{k+1} is reachable in $k$ transitions. Hence $\di{N,I}
\ge k$ and we have a contradiction
\end{proof}

\subsection{Propositions from Section~\ref{sec:sac_props}}

\begin{proposition}
  %\label{prop:lit_vic}
Let $C$ be a clause of a formula $F$. If $F$ is satisfiable, there
exists a clause $C'$ of the $C$-cluster and a literal $l$ of $C'$
shared with $C$ such that the $l$-vicinity of $C'$ has an assignment
satisfying $F$. (So, $l$ is not redundant in $C'$.)
\end{proposition}
\begin{proof}
Denote the $C$-cluster of $F$ as $F_c$.  Let \pnt{p} be an assignment
satisfying $F$. Denote by $\hat{F}_c$ the subset of $F_c$ obtained
from the latter by removing every clause that has a literal \ti{not
  shared} with $C$ and satisfied by \pnt{p}.  (This is either a
literal of a variable that is not in \V{C} or the negation of a
literal of $C$ falsified by \pnt{p}.) Note that the set $\hat{F}_c$ is
not empty since it at least includes the primary clause $C$.  Assume
there is a clause $C' \in \hat{F}_c$ for which \pnt{p} satisfies only
one literal. Then \pnt{p} is the satisfying assignment we look for
because, by definition of $\hat{F}_c$, this literal is shared with
$C$.

Now assume there is no clause $C'$ above. Then every clause of
$\hat{F}_c$ has at least two literals from \slt{C}{p} where \slt{C}{p}
denotes the literals of $C$ satisfied by \pnt{p}.  Below, we show that
by flipping values of \pnt{p} satisfying literals of \slt{C}{p}, one
eventually obtains a required satisfying assignment.  Let $l(x)$ be a
literal of \slt{C}{p} where $x \in X$.  Denote by \pent{p}{^*} the
assignment obtained from \pnt{p} by flipping the value of $x$.  If
there is a clause of $\hat{F}_c$ for which \pent{p}{^*} satisfies only
one literal, then it is a required satisfying assignment. (By
construction, this literal is shared with $C$.) If not, we proceed as
described below.

First, we remove from $\hat{F}_c$ the clauses with the literal
$\overline{l(x)}$. Note that the modified $\hat{F}_c$ is not empty
since it contains at least the primary clause $C$. At this point every
clause of the modified set $\hat{F}_c$ has at least two literals from
$\sslt{C}{p}$. Besides, the number of literals of $C$ satisfied by
\pent{p}{^*} is reduced by one in comparison to \pnt{p} due to
flipping the value of $x$.  Now we pick a new literal $l(x^*)$ from
\sslt{C}{p} and flip the value of $x^*$. Eventually such a procedure
will produce an assignment that satisfies only one literal of either
the primary clause $C$ itself or some other clause from $\hat{F}_c$
\end{proof}
\begin{proposition}
%\label{prop:sac_uns}
Let $C$ be a clause of $F$.  Assume that for every clause $C'$ of the
$C$-cluster, every literal $l$ of $C'$ shared with $C$ is
redundant. Then $F$ is unsatisfiable.
\end{proposition}
\begin{proof}
Assume the contrary \ie $F$ is satisfiable. Then
Proposition~\ref{prop:lit_vic} entails that there exists a clause $C'$
of the $C$-cluster and a literal $l$ of $C'$ shared with $C$ such that
$l$ is not redundant in $C'$. So, we have a contradiction
\end{proof}

\begin{proposition}
  %  \label{prop:optim}
Let $C$ be a clause of a formula $F$ and $F_c$ denote the $C$-cluster.
Let $F'_c$ and $F''_c$ be two disjoint subsets of $F_c$ such that $C
\in F'_c$ and $F_c = F'_c \cup F''_c$. Assume that for every clause
$C' \in F'_c$ and for every literal $l$ of $C'$ shared with $C$ there
is an \lpr{l} implied by $F \setminus F''_c$. Then $F \setminus F''_c$
(and hence $F$) is unsatisfiable.
\end{proposition}
\begin{proof}
Denote by $F'$ the formula $F \setminus F''_c$.  Note that $C$ is in
$F'$ and $F'_c$ is the $C$-cluster for $F'$. By our assumption, for
every clause $C'$ of $F'_c$ and every literal $l$ of $C'$ shared with
$C$, one can derive an \lpr{l} from $F'$. Hence \mprop holds for
$F'$. Then $F'$ (and so $F$) is unsatisfiable
\end{proof}

\subsection{Propositions in support of  Section~\ref{sec:sac_sat_alg}}
\label{ssec:cnfl_cls_props}
\begin{proposition}
\label{prop:cnfl_subsp1}
Let $C$ be a clause of a formula $F(X)$ and \pnt{q} be an assignment
to $X$ that does not satisfy/falsify $C$.  Let $P$ be a set of
certificates.  Suppose that for every clause $C'$ of the $C$-cluster
of $F$ that is not satisfied by \pnt{q} and for every literal $l$ of
$C'$ shared with $C$ and unassigned by \pnt{q}, $P$ contains a clause
that is an \lpr{l} for $C'$ in subspace \pnt{q}. Then $F$ is
unsatisfiable in subspace \pnt{q}.
\end{proposition}
\begin{proof}
The satisfiability of $F$ in subspace \pnt{q} is equivalent to the
satisfiability of \cof{F}{q}.  (Recall that \cof{F}{q} denotes the
formula obtained from $F$ by removing the clauses satisfied by \pnt{q}
and removing the literals falsified by \pnt{q} from the remaining
clauses.)  Note that the \cof{C}{q}\,-cluster for \cof{F}{q} and the
certificates of \cof{P}{q} satisfy the conditions of
Proposition~\ref{prop:sac_uns}. So, \cof{F}{q} is unsatisfiable.
Hence, $F$ is unsatisfiable in subspace \pnt{q}
\end{proof}
\begin{proposition}
\label{prop:cnfl_cls1}
Assume the conditions of Proposition~\ref{prop:cnfl_subsp1} hold for a
formula $F$, a set of certificates $P$, a clause $C \in F$ and an
assignment \pnt{q}.  Then every clause \Sub{B}{ind} satisfying the
three conditions below is falsified by \pnt{q} and implied by $F$.
First, \Sub{B}{ind} includes every literal of $C$ falsified by
\pnt{q}.  Second, for every clause $C'$ of the $C$-cluster satisfied
by \pnt{q}, \Sub{B}{ind} contains the negation of a literal of $C'$
satisfied by \pnt{q}. Third, for every clause $C'$ of the $C$-cluster
unsatisfied by \pnt{q} and for every literal $l$ of $C'$ shared with
$C$ and unassigned by \pnt{q}, \Sub{B}{ind} contains literals of a
clause $B \in P$ that is an \lpr{l} for $C'$ in subspace
\pnt{q}. Namely, \Sub{B}{ind} contains all the literals of $B$ but
those of $C'$.
\end{proposition}
\begin{proof}
The assignment \pnt{q} falsifies \Sub{B}{ind} because the latter
consists only of literals falsified by \pnt{q}. Let \pnt{r} denote the
shortest assignment falsifying the clause \Sub{B}{ind}. By
construction, \pnt{r} is a subset of \pnt{q}. To prove that
\Sub{B}{ind} is implied by $F$, it suffices to show that $F$ is
unsatisfiable in subspace \pnt{r}. Note that $F$ satisfies the
conditions of Proposition~\ref{prop:cnfl_subsp1} in subspace
\pnt{r}. First, $C$ is not satisfied/falsified by \pnt{r}. Second, for
every clause of $C'$ of the $C$-cluster that is not satisfied by
\pnt{r} and for every literal $l$ of $C'$ shared with $C$ and
unassigned by \pnt{r}, there is a clause of $P$ that is an \lpr{l} for
$C'$ in subspace \pnt{r}.  Hence, $F$ is unsatisfiable in subspace
\pnt{r}
\end{proof}

\section{Tests Detecting Stuck-At Faults And PQE}
\label{app:stuck_at}
In this appendix, we show the relation between tests specified by a
property obtained by PQE (as described in Section~\ref{sec:prop_gen})
and those detecting stuck-at faults. Note that tests detecting
stuck-at faults are circuit-specific.  Assume, for instance, that $M'$
and $M''$ are logically equivalent circuits that are structurally
different. Then a set of tests with a very high coverage of stuck-at
faults for $M'$ can have poor coverage for $M''$.

In this appendix, we reuse the notation of
Section~\ref{sec:prop_gen}. Let $M(X,V,W)$ be a combinational circuit
where $X,V,W$ are the internal, input and output variables
respectively.  Let $F(X,V,W)$ be a formula specifying $M$.
Let $g$ be an AND gate of $M$ whose functionality is $x_3 = x_1 \wedge
x_2$. That is $x_1,x_2$ are the input variables of $g$ and $x_3$ is
its output variable.  The functionality of $g$ is specified by the
formula $C_1 \wedge C_2 \wedge C_3$ where $C_1 = \overline{x}_1 \vee
\overline{x}_2 \vee x_3$, $C_2 = x_1 \vee \overline{x}_3$, $C_3 = x_2
\vee \overline{x}_3$ (see Example~\ref{exmp:gate_cnf}). The clauses
$C_1,C_2,C_3$ are present in formula $F$. Consider taking $C_1$ out of
\prob{X}{F}. This clause makes $g$ produce the output value 1 when its
input values are 1.  (If $x_1$ and $x_2$ are set to 1, the clause
$C_1$ can be satisfied only by setting $x_3$ to 1.)

Let $H(V,W)$ be the property obtained by taking out $C_1$. That is
$\prob{X}{F} \equiv H \wedge \prob{X}{F \setminus \s{C_1}}$. Let
$Q(V,W)$ be a clause of $H$. Assume, for the sake of simplicity, that
$H$ does not have clauses implied by $F \setminus \s{C_1}$. (Such
clauses are redundant in $H$ in the sense that $H$ remains a solution
to the PQE problem above even after removing them.)  Then the formula
$(F \setminus \s{C_1}) \wedge \overline{Q}$ is satisfiable. Let
$(\pnt{x}^*,\pnt{v},\pnt{w}^*)$ be a full assignment satisfying this
formula. Note that this assignment \ti{falsifies} $C_1$.  (Indeed,
assume the contrary. Then $(\pnt{x}^*,\pnt{v},\pnt{w}^*)$ satisfies
$F$ because it already satisfies $F \setminus \s{C_1}$. Since this
assignment falsifies $Q$, we have to conclude that $F \not\imp Q$ and
hence\linebreak $F \not\imp H$. So we have a contradiction.)

The fact that $(\pnt{x}^*,\pnt{v},\pnt{w}^*)$ falsifies $C_1$ and
satisfies $F \setminus \s{C_1}$ means that one can view this
assignment as an execution trace of a ``faulty'' version \Sub{M}{flt}
of $M$ specified by $F \setminus \s{C_1}$. So, the fault of
\Sub{M}{flt} is that the output $x_3$ of gate $g$ is \ti{stuck at}
0. (The clause $C_1$ is falsified when $x_1=1,x_2=1,x_3=0$ i.e. if the
gate $g$ outputs 0 when its input variables are assigned 1. So, $g$
outputs 0 for any value assignment to $x_1$ and $x_2$.)

Let $(\pnt{x},\pnt{v},\pnt{w})$ be the execution trace of $M$ under
the input \pnt{v}. Note that the output $\pnt{w}$ is different from
$\pnt{w}^*$ above.  Indeed, since $(\pnt{x},\pnt{v},\pnt{w})$
satisfies $F$ and $F$ implies $Q$, then $(\pnt{v},\pnt{w})$ satisfies
$Q$. Taking into account that $(\pnt{v},\pnt{w}^*)$ falsifies $Q$ one
has to conclude that $\pnt{w} \neq \pnt{w}^*$.  So, the test \pnt{v}
\tb{exposes} the stuck-at fault above because \Sub{M}{flt} can produce
an output that is different from the one produced by $M$ under
\pnt{v}. And since the clause $Q$ is falsified by
$(\pnt{v},\pnt{w}^*)$, one can say that the test \pnt{v} is specified
by $Q$.

%\section{Equivalence Checking By PQE}
%\newpage
\section{Computing \rl{i} By \cp}
\label{app:ec_by_pqe}

In Section~\ref{sec:eq_check}, we recalled~\cp, a method of
equivalence checking by PQE introduced in~\cite{fmcad16}. In this
appendix, we describe how \cp computes the formula \rl{i} specifying
relations between cut points of $\mi{Cut}_i$. We reuse the notation of
Section~\ref{sec:eq_check}.
Let formula $F_i$ specify the gates of $M'$ and $M''$ located between
their inputs and $\mi{Cut}_i$ (see Fig.~\ref{fig:cp_pqe}). Let $Z_i$
denote the variables of $F_i$ minus those of $\mi{Cut}_i$.  Then
\rl{i} is obtained by taking \rl{i-1} out of \prob{Z_i}{\rl{i-1}
  \wedge F_i} \ie $\prob{Z_i}{\rl{i-1} \wedge F_i} \equiv \rl{i}
\wedge \prob{Z_i}{F_i}$.

\setlength{\intextsep}{4pt}
%\begin{wrapfigure}{l}{2.2in}
\begin{figure}[h]
 \begin{center}
    \includegraphics[width=2in]{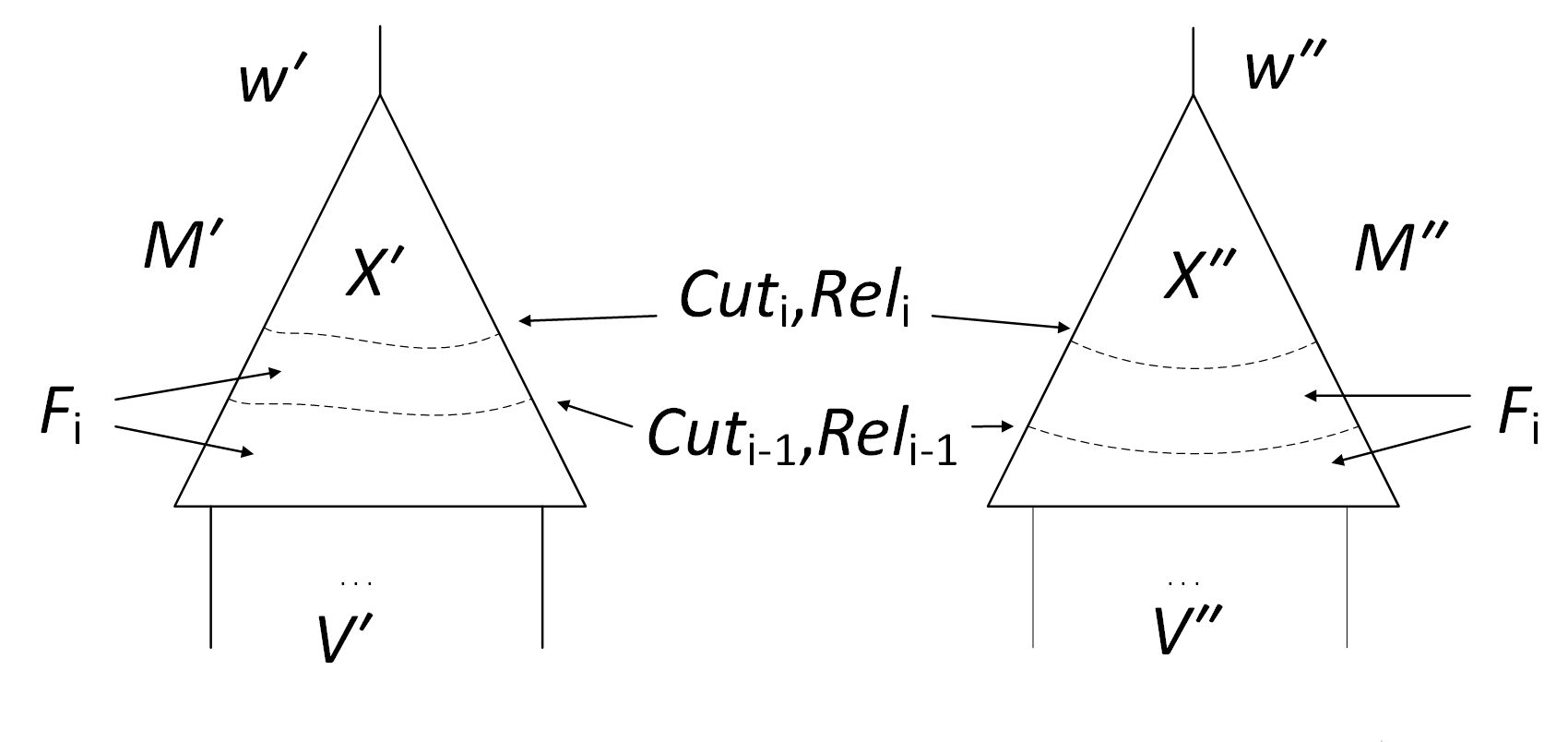}
  \end{center}
%\vspace{-10pt}
\caption{Computing \rl{i} in \cp}
\vspace{3pt}
\label{fig:cp_pqe}
\end{figure}
%\end{wrapfigure}

The formula \rl{i} depends only on variables
of $\mi{Cut}_i$. (The other variables of $\rl{i-1} \wedge F_i$ are in
$Z_i$ and hence, quantified.)
Since \rl{i} is obtained by taking out \rl{i-1}, the latter is
redundant in \prob{Z_i}{\rl{i-1} \wedge \rl{i} \wedge F_i}. One can
show that this property \ti{implies} the property mentioned in
Section~\ref{sec:eq_check}: \rl{i-1} is redundant in \prob{Z}{\rl{i-1}
  \wedge \rl{i} \wedge F}. The latter is just a \ti{weaker} version of
the former.

% Plan
%  *) reason why current PQE solvers have to add conflict clauses
%  *) description of \sapqe
%     +) main operation and its difference from \sasat
%     +) three differences from current PQE solvers
%        -) no variable splitting on quantified variables
%        -) induction allows one to keep formula intact,
%           which address the flaw mentioned above

\section{Using Technique Of \sasat To Fix PQE Solving}
\label{app:sac_pqe}
In this appendix, we give an idea of how the problem with PQE solving
mentioned in Section~\ref{sec:pqe_flaws} can be fixed using the
technique of \sasat. We will refer to the ``fixed'' PQE solver as
\sapqe.  Consider the PQE problem of taking $G$ out of
\prob{X}{F(X,Y)} where $G \in F$.  We will use \dpqe as a
representative of existing PQE solvers.  \dpqe maintains the following
invariant: once a subspace \pnt{q} is reached via variable splitting,
the clauses of $G$ and the current local target clauses must be
proved/made redundant in \prob{X}{F} (see
Subsection~\ref{ssec:targets}). So, if a conflict occurs in subspace
\pnt{q}, a conflict clause \ti{has to} be added to $F$ to
\ti{guarantee} that the target clauses are redundant in \prob{X}{F} in
this subspace.

The difference of \sapqe from \dpqe is threefold. First, \sapqe
replaces splitting on a quantified variable with exploring an
$l(x)$-vicinity of a clause $C$ where $x \in X$, to check the
redundancy of the literal $l(x)$ in $C$. (However, \sapqe still splits
on variables of $Y$. Like in \dpqe, the unquantified variables are
assigned first.) Second, \sapqe adds a learned clause $B$ to $F$
\ti{only} if $B$ is unquantified \ie $\V{B} \subseteq Y$. If $B$ is
quantified, it is added to a set $P$ of clauses certifying literal
redundancy like it is done in \sasat. This means that \sapqe maintains
the invariant above that every target clause must be proved/made
redundant in the current subspace \pnt{q} \ti{only} if $\Va{q}
\subseteq Y$. (Like for \dpqe, the unquantified clauses added to $F$,
form a solution to the PQE problem at hand.)  Third, \sapqe benefits
from checking literal redundancy by applying induction enabled by
\mprop. Using induction is precisely the reason why \sapqe can avoid
adding quantified clauses to $F$.

There is also difference between checking literal redundancy in \sapqe
and \sasat. It is twofold.  First, the redundancy of $l(x)$ in a
clause $C$ is checked by \sapqe in a subspace \pnt{q} only if $C$ is a
target clause (see Subsection~\ref{ssec:targets}). Second, \sapqe
operates differently from \sasat when the former fails to prove $l(x)$
redundant.  Namely, \sasat shows that $l(x)$ is not redundant in
subspace \pnt{q} by finding a satisfying assignment whereas \sapqe
simply proves $C$ blocked at $x$ in this subspace. The latter means
that $C$ is redundant in \prob{X}{F} in subspace \pnt{q} and thus the
redundancy of literal $l(x)$ in $C$ in this subspace is
\ti{irrelevant}.

More specifically, to prove $l(x)$ redundant in \prob{X}{F} in
subspace \pnt{q}, \sapqe checks literal redundancy in the clauses of
$F$ resolvable on $x$ with $C$ in subspace \pnt{q}. If the set of such
clauses is empty, $C$ is trivially blocked at $x$ in subspace \pnt{q}.
Otherwise, \sapqe either derives a certificate clause proving $l(x)$
redundant in $C$ in subspace \pnt{q} or shows that every clause of $F$
resolvable with $C$ on $x$ is redundant in \prob{X}{F} in subspace
\pnt{q}. Then $C$ is blocked at $x$ in this subspace.

\section{An Example Of How \imsat Operates}
\label{app:sac_sat_exmpl}
%\stepcounter{example}
%\input{s10at_exmpl.fig}
%\ti{Example \theexample}.
%\vspace{-5pt}
 
In this appendix, we show how \imsat proves a formula
unsatisfiable. Let $F\!=\!C_1 \wedge \dots \wedge C_9$ where
$C_1\!=\!x_1 \vee x_2$, $C_2 = x_1 \vee x_3$, $C_3 = x_2 \vee x_4$,
$C_4 = \overline{x}_1 \vee x_3$, $C_5 = \overline{x}_2 \vee x_4$, $C_6
= \overline{x}_1 \vee x_5$, $C_7 = \overline{x}_5 \vee
\overline{x}_4$, $C_8 = \overline{x}_2 \vee x_6$, $C_9 =
\overline{x}_6 \vee \overline{x}_3$.
Note that $C_1,C_2,C_3$ are the only clauses of $F$ having literals
$x_1$ or $x_2$ \ie those of \V{C_1}. So, they form the $C_1$-cluster.

The execution trace of a run of \imsat on $F$ is shown in
Fig.~\ref{fig:prot}. The pseudocode of \imsat is presented in
Fig.~\ref{fig:top_proc_sat} of Section~\ref{sec:sac_sat_alg}.
In the initial call of \imsat, the set of certificates $P$, the
current assignment \pnt{q} to $X$ are empty. The clause $C$ and the
literal $l$ are $\mi{nil}$, which means that initially there is no
$l$-vicinity to examine (line 1 of Fig.~\ref{fig:prot}). (The
assignment \pnt{q} stays empty for the entire execution trace of
Fig.~\ref{fig:prot}.  The formula $F$ does not change either whereas
the set of certificates $P$ grows.)  Then \imsat runs \ti{BCP} (line
2). Since $F \cup P$ does not contain any unit clauses and
$C=\mi{nil}$, $l=\mi{nil}$, no assignment is made by \ti{BCP}. After
that, \imsat starts a \ti{while} loop.

%\newpage
%
% \dpqe procedure
%
\setlength{\intextsep}{4pt}
\setlength{\textfloatsep}{2pt}
\begin{wrapfigure}{L}{1.7in}
%\begin{figure}
%\begin{center}
%\small
\footnotesize
%\normalsize
%\vspace{5pt}
\begin{tabbing}
aa\=bb\= cc\= dd\= \kill
%\scriptsize{1}\> $P = \emptyset$, $\pnt{q} = \emptyset$ \\
// $\xi$ denotes $F \wedge P,\pnt{q}$\\
// initial call: $C\!=\!\mi{nil},l\!=\!\mi{nil}$ \\
// $\pnt{q}=\emptyset, F = \Sub{F}{init}$, lines 1-24 \\
// $P = \emptyset$, lines 1-6 \\
// $P = \s{B'_1}$, lines 8-12 \\[1pt]
// $P = \s{B'_1,B''_1}$, lines 14-18  \\[1pt]
// $P = \s{B'_1,B''_1,C_4}$, lines 20-24 \\[1pt]
// $P = \s{B'_1,B''_1,C_4,C_5}$, lines 26-28\\
// \\
\scriptsize{1}\> $\imsat(\Sub{F},P,\pnt{q},C,l)$\\[2pt]
\scriptsize{2}\Tt $\mi{BCP}(\xi,C,l) \Rightarrow \mi{no\,confl}$\\
\scriptsize{3} \ttt $C_1,x_1 := \mi{PickClsLit}(F,P,\pnt{q})$\\
\scriptsize{4}\ttt $\imsat(F,P,\pnt{q},C_1,x_1)$ \\[2pt]
\scriptsize{5}\tttt  $\mi{BCP}(\xi,C_1,x_1) \Rightarrow \mi{confl}$  \\
\scriptsize{6}\tttt  $B'_1 := \mi{GenByRes}(\xi,C_1,x_1)$ \\
\scriptsize{7}\ttt $P := P \cup \s{B'_1}$ \\
\scriptsize{8}\ttt  $\mi{Induct}(F,P,\!\pnt{q})$ returns $\mi{nil}$ \\
~~~~~~~~~$----------$ \\
\scriptsize{9} \ttt $C_1,x_2 := \mi{PickClsLit}(F,P,\pnt{q})$\\
\scriptsize{10}\ttt $\imsat(F,P,\pnt{q},C_1,x_2)$\\[2pt]
\scriptsize{11}\tttt  $\mi{BCP}(\xi,C_1,x_2) \Rightarrow \mi{confl}$  \\
\scriptsize{12}\tttt  $B''_1 := \mi{GenByRes}(\xi,C_1,x_2)$ \\
\scriptsize{13}\ttt $P := P \cup \s{B''_1}$ \\
\scriptsize{14}\ttt  $\mi{Induct}(F,P,\pnt{q})$ returns $\mi{nil}$ \\
~~~~~~~~~$----------$ \\
\scriptsize{15} \ttt $C_2,x_1 := \mi{PickClsLit}(F,P,\pnt{q})$\\
\scriptsize{16}\ttt $\imsat(F,P,\pnt{q},C_2,x_1)$\\[2pt]
\scriptsize{17}\tttt  $\mi{BCP}(\xi,C_2,x_1)\Rightarrow \mi{confl}$  \\
\scriptsize{18}\tttt  $C_4 := \mi{GenByRes}(\xi,C_2,x_1)$ \\
\scriptsize{19}\ttt $P := P \cup \s{C_4}$ \\
\scriptsize{20}\ttt  $\mi{Induct}(F,P,\!\pnt{q})$ returns $\mi{nil}$ \\
~~~~~~~~~$----------$ \\
\scriptsize{21} \ttt $C_3,x_2 := \mi{PickClsLit}(F,P,\pnt{q})$\\
\scriptsize{22}\ttt $\imsat(F,P,\pnt{q},C_3,x_2)$\\[2pt]
\scriptsize{23}\tttt  $\mi{BCP}(\xi,C_3,x_2)\Rightarrow \mi{confl}$  \\
\scriptsize{24}\tttt  $C_5 := \mi{GenByRes}(\xi,C_3,x_2)$ \\
\scriptsize{25}\ttt $P := P \cup \s{C_5}$ \\
\scriptsize{26}\ttt  $\mi{Induct}(F,P,\pnt{q})$ returns $C_1$ \\
%~~~~~~~~~$----------$ \\
\scriptsize{27}\tttt  $\Sub{B}{ind}\!:=\!\mi{FormCls}(F,P,C_1,\pnt{q})$ \\
\scriptsize{28}\ttt   return($P,\Sub{B}{ind}$) \\
%\tb{\scriptsize{1}}\>   \\
\end{tabbing} 
\vspace{-15pt}
\caption{Execution trace of \imsat}
\vspace{6pt}
\label{fig:prot}
%\end{figure}
\end{wrapfigure}

In the \ti{while} loop, \imsat first calls \ti{PickClsLit} to pick a
clause of $F$ and a literal of this clause (line 3). Assume the clause
$C_1$ and its literal $x_1$ are picked.  To check the redundancy of
$x_1$ in $C_1$, \imsat recursively calls itself to explore the
$x_1$-vicinity of $C_1$ \ie the subspace
\mbox{$(\bm{x_1\!=\!1},x_2\!=\!0)$} (line 4). The new activation of
\imsat calls \ti{BCP} to make the assignment above to $x_1$ and $x_2$
(line 5).  To satisfy the clause $C_3$ that becomes unit when $x_2=0$,
\ti{BCP} makes the assignment $x_4=1$. The clause $C_6$ becomes unit
when $x_1=1$ and the assignment $x_5=1$ is made by \ti{BCP}. This
leads to falsifying all literals of the clause $C_7$ \ie to a
conflict. By resolving clauses $C_6,C_7$ and $C_3$, \ti{GenByRes}
produces the conflict clause $B'_1 = \overline{x}_1 \vee x_2$ that
also serves as an \lpr{x_1} for $C_1$ (line 6). After that, the
recursive invocation of \imsat terminates returning $B'_1$. The latter
is added to the set of certificates $P$ (line 7). Finally, \imsat
calls \ti{Induct} to check if proving redundancy of $x_1$ in $C_1$
enabled induction (line 8). Namely, whether there is a clause $C'$ of
$F$ such that \mprop holds for the $C'$-cluster. \ti{Induct} returns
\ti{nil} indicating that it failed to find such a clause.

Now assume that in the second iteration of the \ti{while} loop,
\ti{PickClsLit} picks the clause $C_1$ and the literal $x_2$ (line
9). Then the actions similar to those of the first iteration are
performed. Namely, \imsat checks the redundancy of $x_2$ in $C_1$ by
recursively calling itself to examine the $x_2$-vicinity of $C_1$ \ie
the subspace $(x_1=0,\bm{x_2=1})$ (line 10).  This new invocation of
\imsat runs \ti{BCP}, which results in a conflict (line 11). By
resolving clauses $C_8,C_9$ and $C_2$, \ti{GenByRes} produces the
conflict clause $B''_2=x_1 \vee \overline{x}_2$ that also serves as an
\lpr{x_2} for $C_1$. The clause $B''_2$ is added to $P$. Finally,
\ti{Induct} returns \ti{nil} since proving redundancy of $x_2$ in
$C_1$ does not enable induction yet.

Assume that in the third iteration of the \ti{while} loop,
\ti{PickClsLit} picks the clause $C_2$ and the literal $x_1$ (line
15).  Then the actions similar to those of the first two iterations
are performed.  Namely, \imsat checks the redundancy of $x_1$ in $C_2$
by recursively calling itself to examine the $x_1$-vicinity of $C_2$
\ie the subspace $(\bm{x_1=1},x_3=0)$ (line 16). The only difference
is as follows. When the new invocation of \imsat runs BCP, a conflict
involving only the clauses $C_2$ and $C_4$ occurs after assigning
$x_3=0$ (line 17). This means that $C_4 = \overline{x}_1 \vee x_3$ is
an \lpr{x_1} for $C_2 = x_1 \vee x_3$ and there is no need to produce
a \ti{new} certificate clause by resolution (line 18). So, $C_4$ is
added to $P$ (line 19).  Finally, \ti{Induct} returns \ti{nil}
indicating that proving redundancy of $x_1$ in $C_2$ does not enable
induction yet (line 20).

Assume that in the fourth iteration of the \ti{while} loop,
\ti{PickClsLit} picks the clause $C_3$ and the literal $x_2$ (line
21). This iteration is similar to the previous one.  Namely, when
exploring the $x_2$-vicinity of $C_3$ by a new invocation of \imsat,
\ti{BCP} runs into a conflict involving only $C_3$ and $C_5$. This
means that $C_5=\overline{x}_2 \vee x_4$ is an \lpr{x_2} for $C_3=x_2
\vee x_4$ (lines 22-24). So, $C_5$ is added to $P$ (line 25).  The
difference from the three previous iterations however is that
\ti{Induct} returns $C_1$ indicating that \mprop holds for the
$C_1$-cluster in subspace $\pnt{q} = \emptyset$ (line 26). Then an
empty clause \Sub{B}{ind} is generated indicating that $F$ is
unsatisfiable (line 27) and the original call of \imsat terminates
(line 28). \imsat returns the empty clause \Sub{B}{ind} and the set of
certificates $P$ equal to \s{B'_1,B''_1,C_4,C_5}.

Note that the formula $P$ specified by the set of certificates
\s{B'_1,B''_1,C_4,C_5} is satisfiable. It remains satisfiable even if
conjoined with $F_{c_1}$ where $F_{c_1}$ denotes the $C_1$-cluster \ie
\s{C_1,C_2,C_3}.  So the induction step is non-trivial in the sense
that an empty clause cannot be obtained from $P \wedge F_{c_1}$ by
resolution. Note also that $P \wedge F_{c_1}$ becomes unsatisfiable if
the \ti{blocked} clause $\overline{x}_1 \vee \overline{x}_2 \vee
\overline{x}_3$ is added. (It is blocked at $x_1$.)

\section{Identifying Literal Redundancy By BCP}
\label{app:red_by_bcp}
In Example~\ref{exmp:cnfl_cls} of Section~\ref{sec:sac_sat_alg}, we
showed how \imsat proves redundancy of a literal if a conflict occurs
in BCP. In this appendix, we describe two cases where \imsat proves
redundancy of a literal during BCP \ti{without} reaching a conflict.

Let $C$ be a clause of a formula $F$ and $l$ be a literal of
$C$. Suppose \imsat checks the redundancy of the literal $l$ in
subspace \pnt{q}. Example~\ref{exmp:sat_lit} shows how \imsat proves
$l$ redundant in subspace \pnt{q} if a literal of $C$ \ti{different}
from $l$ is \ti{satisfied} by an assignment derived during BCP.
Example~\ref{exmp:fls_lit} demonstrates that the same can be done if
an assignment \ti{falsifying} the literal $l$ \ti{itself} is derived
during BCP.

\begin{example}
  \label{exmp:sat_lit}
Let $F=C_1 \wedge C_2 \wedge C_3 \wedge \ldots$ where $C_1 = x_1 \vee
x_2 \vee x_3$, $C_2 = x_1 \vee x_4$, $C_3 = x_2 \vee
\overline{x}_4$. Suppose \imsat needs to prove redundancy of the
literal $x_3$ in $C_1$ in subspace \pnt{q}. Assume, for the sake of
simplicity, that $\pnt{q} = \emptyset$. The $x_3$-vicinity of $C_1$ is
specified by the assignment\linebreak
$(x_1=0,x_2=0,\bm{x_3=1})$. Suppose, \imsat first makes the decision
assignment $x_1=0$ and then runs BCP.  Then the assignment $x_4=1$ is
derived from $C_2$ and $x_2=1$ is derived from $C_3$. Note that the
latter assignment \ti{satisfies} a literal of $C_1$ different from
$x_3$ (\ie the literal $x_2$). This means that $x_3$ is proved
redundant in $C_1$ before any conflict occurs. The clause $x_1 \vee
x_2$ obtained by resolving $C_2$ and $C_3$ on $x_4$ is an \lpr{x_3}
for $C_1$.
\end{example}

\begin{example}
  \label{exmp:fls_lit}
Consider a different example where $C_3$ is replaced with $C^*_3$. Let
$F= C_1 \wedge C_2 \wedge C^*_3 \wedge \ldots$ where $C_1 = x_1
\vee\!x_2\!\vee x_3$, $C_2\!=\!x_1\!\vee\!x_4$,
$C^*_3\!=\!\overline{x}_3\!\vee\!\overline{x}_4$. Suppose again that
\imsat needs to prove redundancy of the literal $x_3$ in $C_1$ in
subspace $\pnt{q} = \emptyset$. As before, the $x_3$-vicinity of $C_1$
is specified by the assignment\linebreak
$(x_1=0,x_2=0,\bm{x_3=1})$. If \imsat makes the decision assignment
$x_1=0$, BCP derives the assignment $x_4\!=\!1$ from $C_2$ and
$x_3\!=\!0$ from $C^*_3$. Note that the latter assignment
\ti{falsifies} the literal $x_3$ of $C_1$. This means that $x_3$ is
proved redundant in $C_1$ without reaching a conflict. The clause $x_1
\vee \overline{x}_3$ obtained by resolving $C_2$ and $C^*_3$ on $x_4$
is an \lpr{x_3} for $C_1$.
\end{example}

\end{document}